\documentclass[a4paper,reqno]{amsart}

\usepackage[margin=2cm,papersize={19.75cm,25cm}]{geometry}
\usepackage[all]{xy}           
\usepackage{amssymb}           
\usepackage{hyperref}
\usepackage{eucal}
\usepackage[dvips]{graphics}
\usepackage{graphicx}
\usepackage{exscale}
\usepackage{color}
\usepackage{tikz}
\usetikzlibrary{matrix}
\usepackage{subfig}
\DeclareGraphicsExtensions{.jpg,.png,.pdf}
\usepackage{exscale}
\usepackage{pgfplots}
\usepackage{tikz}
\usetikzlibrary{arrows,chains,matrix,positioning,scopes}
\makeatletter
\tikzset{join/.code=\tikzset{after node path={%
\ifx\tikzchainprevious\pgfutil@empty\else(\tikzchainprevious)%
edge[every join]#1(\tikzchaincurrent)\fi}}}
\makeatother
\tikzset{>=stealth',every on chain/.append style={join},
         every join/.style={->}}
\tikzstyle{labeled}=[execute at begin node=$\scriptstyle,
   execute at end node=$]

\numberwithin{equation}{section}

\newtheorem{definition}{Definition}[section]

\newtheorem{proposition}[definition]{Proposition}
\newtheorem{corollary}[definition]{Corollary}
\newtheorem{remarkth}[definition]{Remark}

\newenvironment{remark}{\begin{remarkth}\upshape}{\hfill$\diamond$\end{remarkth}}

\renewcommand{\emph}[1]{{\bfseries\itshape{#1}}}

\newcommand{\lcf}{\lbrack\! \lbrack}
\newcommand{\rcf}{\rbrack\! \rbrack}

\newcommand{\lvec}[1]{\overleftarrow{#1}}
\newcommand{\rvec}[1]{\overrightarrow{#1}}

\makeatletter
\newcommand\prol{\@ifstar{\@proldf}{\@prolpf}}  
\def\@prolpf{\@ifnextchar[{\@prolpf@wrt}{\@prolpf@}}
\def\@prolpf@wrt[#1]#2{\@ifnextchar[{\@prolpf@wrt@at{#1}{#2}}{\@prolpf@wrt@{#1}{#2}}}
\def\@prolpf@wrt@at#1#2[#3]{\prolsymbol^{#1}_{#3}#2}
\def\@prolpf@wrt@#1#2{\prolsymbol^{#1}#2}
\def\@prolpf@#1{\@ifnextchar[{\@prolpf@at{#1}}{\@prolpf@@{#1}}}
\def\@prolpf@at#1[#2]{\prolsymbol_{#2}#1}
\def\@prolpf@@#1{\prolsymbol#1}
\def\@proldf{\@ifnextchar[{\@proldf@wrt}{\@proldf@}}
\def\@proldf@wrt[#1]#2{\@ifnextchar[{\@proldf@wrt@at{#1}{#2}}{\@proldf@wrt@{#1}{#2}}}
\def\@proldf@wrt@at#1#2[#3]{\prolsymbol^{*#1}_{#3}#2}
\def\@proldf@wrt@#1#2{\prolsymbol^{*#1}#2}
\def\@proldf@#1{\@ifnextchar[{\@proldf@at{#1}}{\@proldf@@{#1}}}
\def\@proldf@at#1[#2]{\prolsymbol^*_{#2}#1}
\def\@proldf@@#1{\prolsymbol^*#1}
\def\prolsymbol{\mathcal{T}}
\makeatother

\def\lcf{\lbrack\! \lbrack}
\def\rcf{\rbrack\! \rbrack}
\def\lcf{\lbrack\! \lbrack}
\def\rcf{\rbrack\! \rbrack}

\begin{document}

\title[ Invariant Poisson-Nijenhuis structures on Lie groups and classification]{Invariant Poisson-Nijenhuis structures on Lie groups and classification}

\author[Z. Ravanpak]{Zohreh Ravanpak}
\address{Z.\ Ravanpak:
	Azarbaijan Shahid Madani University \\
	Department of Mathematics, Faculty of Science \\
	Tabriz, Iran, 
  and Institute of Mathematics\\ Polish Academy of Sciences\\
Warsaw, Poland}
\email{z.ravanpak@azaruniv.ac.ir \& zravanpak@impan.pl}

\author[A. Rezaei-Aghdam]{Adel Rezaei-Aghdam}
\address{A.\ Rezaei-Aghdam:
	Azarbaijan Shahid Madani University \\
	Department of Physics, Faculty of Science \\
	Tabriz, Iran}
\email{rezaei-a@azaruniv.ac.ir}

\author[Gh. Haghighatdoost]{Ghorbanali Haghighatdoost}
\address{Gh.\ Haghighatdoost:
	Azarbaijan Shahid Madani University \\
	Department of Mathematics, Faculty of Science \\
	Tabriz, Iran}
\email{gorbanali@azaruniv.ac.ir}
\keywords{Poisson-Nijenhuis structures, Lie bialgebras and coboundary Lie bialgebras, Integrable dynamical systems}

\subjclass[2010]{37K05, 37K10, 53D17, 37K30 }
\thanks{This research was supported by research fund No. $217.d.14312$ from Azarbaijan Shahid Madani University. The support from the Warsaw Center for Mathematics and Computer Science is also acknowledged.
}

\begin{abstract}
 We study {\em right-invariant (resp., left-invariant)  Poisson-Nijenhuis structures} on a Lie group $G$ and introduce their infinitesimal counterpart, the so-called {\em r-n structures} on the corresponding Lie algebra $\mathfrak g$. We show that $r$-$n$ structures can be used to find compatible solutions of the classical Yang-Baxter equation. Conversely, two compatible r-matrices from which one is invertible determine an $r$-$n$ structure.
 We classify, up to a natural equivalence, all $r$-matrices and all $r$-$n$ structures with invertible $r$ on four-dimensional symplectic real Lie algebras. The result is applied to show that a number of dynamical systems which can be constructed by $r$-matrices on a phase space whose symmetry group is Lie group $G$, can be specifically determined.
\end{abstract}

\maketitle

\date{\today}

\section{Introduction}
Poisson-Nijenhuis ($P$-$N$) structures on manifolds were introduced by Magri and Morosi \cite{MaMo} (see also \cite{KoMa}), as Poisson and Nijenhuis structures which are compatible in a sense.

In this work we study Poisson-Nijenhuis structures on a Lie group $G$ which are appropriate right-invariant (or resp. left-invariant), so-called {\em right-invariant $P$-$N$ structures} (or resp. {\em left-invariant $P$-$N$ structures}) on $G$; we introduce their infinitesimal counterpart, the objects which we called {\em r-n structures} on the Lie algebra $\mathfrak g$ of $G$.

Actually, $r$-matrices, i.e. solutions of the classical Yang-Baxter equation (CYBE) on $\mathfrak g$, are in one-to-one correspondence with right-invariant (resp. left-invariant) Poisson structures on $G$.

We discus the relation between the infinitesimal version of right-invariant $P$-$N$ structures on the Lie group $G$ and coboundary Lie bialgebras on the Lie algebra $\mathfrak g$ of $G$. We show that $r$-$n$ structures can produce compatible solutions of the CYBE. Conversely, we show that two $r$-matrices under a certain condition are related by $r$-$n$ structures on $\mathfrak g$.

For clarity of our results, we classify all $r$-$n$ structures on a specific class of Lie algebras. We consider four-dimensional symplectic real Lie algebras, given in \cite{GOv}; and classify, up to a natural equivalence, all r-matrices on them; then we classify, all $r$-$n$ structures with invertible $r$ on considered Lie algebras. We did all computations with the use of Maple mathematical software.

In general, one can find a classification of invertible $r$-matrices and then classify all pairs of compatible $r$-matrices with the invertible one. The method is based on two following known observations:

\begin{itemize}
	\item Invariant structures on the Lie group can be identified with their infinitesimal part.
	\item Given an invertible Poisson structure, one can find compatible Poisson structures with the given one using the theory of $P$-$N$ structures.
\end{itemize}
Solutions of CYBE are closely related to the theory of classical integrable systems. A method for constructing integrable models using the classical $r$-matrices has been introduced in \cite{Zh}. Recently, a generalization of the relation between the classical Yang-Baxter equation and the theory of classical integrable system has been considered in \cite{ReSe}; the authors presented a method to construct classical integrable systems using coboundary Jacobi-Lie bialgebras.

 This is a good motivation for discussing the particular class of $P$-$N$ structures on a Lie group $G$, right-invariant $P$-$N$ structures, since the infinitesimal counterpart of them, $r$-$n$ structures, are really good tools to construct the classical integrable systems using the procedure mentioned in \cite{Zh}. An interesting point is that a number of dynamical systems which can be constructed by $r$-matrices on a phase space whose symmetry group is Lie group $G$ can be specifically determined.

The outline of the paper is as follows: In Section \ref{Section2} we briefly recall the notion of $P$-$N$ structures on a manifold, Lie bialgebras and coboundary Lie bialgebras, as needed in the rest of the paper. In Section \ref{Section3} we define right-invariant $P$-$N$ structures on the Lie group $G$ as the main object of study and introduce their infinitesimal counterpart. In Section \ref{Section4} we show that compatible $r$-$n$ structures, under some conditions, are in one-to-one correspondence with compatible $r$-matrices. The classification procedure of right-invariant $P$-$N$ structures is the subject of Section \ref{method}. In Section \ref{cllasification} we list the results of a classification of $r$-$n$ structures on four-dimensional symplectic real Lie algebras; in order to do, in \ref{class-r-matrices} we classify, up to a natural equivalence, all $r$-matrices on those Lie algebras; in \ref{all-r-n} all $r$-$n$ structures with invertible $r$ are given. Finally, in \ref{class-r-n} we classify, up to a natural equivalence, all Nijenhuis structures corresponding to the given $r$-$n$ structures. In Section \ref{app} we give an example to clarify the application of our results in previous sections. We end the paper by a concluding section.

 \section{Antecedents }\label{Section2}
 In this section, we recall the definition of Poisson-Nijenhuis
 structures \cite{KoMa}. We will very briefly review the notion of Lie bialgebra and coboundary Lie bialgebra. 

 \subsection{Poisson-Nijenhuis structure}
 A Poisson-Nijenhuis ($P$-$N$)  structure on a manifold $M$ is a bivector field $\Pi: T^{*}M \times T^{*}M \to \mathbb R $ together with a $(1,1)$-tensor field $N: TM \to TM$  on $M$ satisfying the conditions:
  \begin{enumerate}
  \item $\Pi$ is a Poisson bivector, i.e.
  \[
   [\Pi, \Pi] = 0,
   \]
  \noindent  where $[\cdot, \cdot ]$ is the Schouten-Nijenhuis bracket.
  \item  $N$ is a Nijenhuis operator, i.e.
  \[
  [N,N]\footnote{ The bracket $[\cdot,\cdot]$ is the Fr\"{o}licher-Nijenhuis bracket, that is the bracket of vector-valued differential forms.} (X,Y) = [NX,NY] - N[NX, Y] - N[X,NY] + N^2[X, Y]=0, \quad \forall X, Y \in {\mathfrak X}(M).
  \]
  \end{enumerate}
\noindent Moreover, $\Pi$ and $N$ satisfy these two compatibility conditions:
  \begin{enumerate}
  \item
   \[
   N \circ \Pi^{\sharp} = \Pi^{\sharp} \circ N^t,
   \]
where $\Pi^{\sharp}$ is induced by $\Pi$ as a vector bundle map from the cotangent bundle $T^*M$ of $M$ on the tangent bundle $TM$ of $M$, $\Pi^{\sharp}: T^*M \to TM$, given by
   \[
   \left\langle \alpha, \Pi^{\sharp}(\beta) \right\rangle  = \Pi(\alpha, \beta), \; \; \mbox{for}\;\alpha, \beta \in T_x^*M \; \mbox{and}\; x\in M.
   \]
   \item
\begin{equation}\label{Def-Concomi}
    \begin{array}{rcl}
     C(\Pi,N)(\alpha,\beta)&=& \mathcal{L}_{\Pi^{\sharp}\alpha}(N^{t}\beta)-\mathcal{L}_{\Pi^{\sharp}\beta}(N^{t}\alpha)+N^{t}\mathcal{L}_{\Pi^{\sharp}\beta}\alpha-N^{t}\mathcal{L}_{\Pi^{\sharp}\alpha}\beta \\[4pt]
     &&+d \left\langle \alpha,N\Pi^{\sharp}\beta\right\rangle
     -N^{t}d\left\langle \alpha,\Pi^{\sharp}\beta\right\rangle=0,
     \end{array}
  \end{equation}
     for $\alpha, \beta \in \Omega^1(M)$. Here, $N^{t}:T^{*}M \to T^{*}M$ is the dual $(1,1)$-tensor field to $N$ and $C(\Pi, N)$ is a $(2,1)$-tensor field on $M$, a concomitant of $\Pi$ and $N$.
\end{enumerate}

\noindent Note that $C(\Pi,N)(\alpha,\beta)=\left\lbrace \alpha,\beta\right\rbrace_{N\Pi}- \left\lbrace \alpha,\beta\right\rbrace^{N^{t}}_\Pi\,$ where
	\[
	\begin{array}{rcl}
\left\lbrace \alpha,\beta\right\rbrace^{N^{t}}_\Pi
		&=&\left\lbrace N^{t}\alpha,\beta\right\rbrace_{\Pi}+\left\lbrace \alpha,N^{t}\beta \right\rbrace_{\Pi} -N^{t}\left\lbrace \alpha,\beta \right\rbrace_{\Pi},\\[4pt]
		\end{array}
		\]
and the bracket $\{\cdot,\cdot\}_{\Pi}$ is the bracket of 1-forms which is defined by the Poisson bivector $\Pi$ as follows:
\begin{equation}\label{Poisson-bracket}
\left\lbrace \alpha,\beta\right\rbrace _{\Pi}  =\mathcal L_{\Pi^{\sharp}\alpha}\beta-\mathcal L_{\Pi^{\sharp}\beta}\alpha+d(\Pi(\alpha,\beta)).
\end{equation}	
Similarly, the bracket $\left\lbrace \alpha,\beta\right\rbrace_{N\Pi}$ is the bracket of 1-forms defined by the 2-contravariant tensor $N\circ \Pi$ (for more details see, \cite{KoMa}).

  An important fact \cite{MaMo} is that in the presence of a $P$-$N$ structure $(\Pi, N)$ on $M$ one may produce a hierarchy of Poisson structures $\Pi_k$, $k \in \mathbb{N} \cup \{0\}$, which are compatible, that is,
  \[
  [\Pi_k, \Pi_l] = 0, \; \; \mbox{ for } k, l \in \mathbb{N}.
  \]
  The Poisson structure $\Pi_k$ in this hierarchy is characterized by the condition
  \[
  \Pi_k^{\sharp} = N^k \circ \Pi^{\sharp}, \; \; \mbox{ for } k \in \mathbb{N}, \; \; k \geq 0.
  \]
  This follows from the fact that if $(\Pi,N)$ is a $P$-$N$ structure, then the couple $(\Pi, N^k)$ is also $P$-$N$ structure.

  Note that, Poisson-Nijenhuis structures are closely related with the theory of completely integrable systems (see \cite{KoMa,MaMo}).

   \subsection{Lie bialgebras and Coboundary Lie bialgebras } Before we recall Lie bialgebras, we have to recall a few definitions from the theory of the Lie algebra cohomology. Any Lie algebra $\mathfrak g$ acts on itself by the adjoint representation $ad:X\in\mathfrak g \to ad_X\in\mbox{End}( \mathfrak g)$, defined by $ad_{X}Y=[X,Y]$, for $Y\in \mathfrak g$.
   In general, $\mathfrak g$ acts on the space $\otimes^p\mathfrak g=\mathfrak g\otimes...\otimes\mathfrak g$ of tensor product of $\mathfrak g$ by the {\em generalized adjoint representation} of the Lie algebra $\mathfrak g$, in the following way,
   \[
   \begin{array}{rcl}
  X.(Y_1\otimes...\otimes Y_p)&=& ad_X^{(p)}(Y_1\otimes...\otimes Y_p) \\
  &=& ad_X Y_1\otimes...\otimes Y_p+Y_1\otimes ad_X Y_2\otimes ...\otimes Y_p+...+Y_1\otimes Y_2\otimes...\otimes ad_X Y_p.
  \end{array}
   \]

   \noindent  Similarly, $\mathfrak g$ acts on $\bigwedge^2\mathfrak g$ 
  by
   \begin{equation}\label{ad-rep}
   X.(Y_1\wedge Y_2)=[X,Y_1]\wedge Y_2+Y_1\wedge[X,Y_2].
   \end{equation}
A $k$-linear skew-symmetric map from $\mathfrak g$ to a vector space $V$ which is a {\em $\mathfrak g$-module} (a vector space in which we have a representation of the Lie algebra $\mathfrak g$), is a {\em $k$-cochain} of $\mathfrak g$ with values in $V$, so a linear map from $\mathfrak g$ to $\bigwedge ^2 \mathfrak g$ is a $1$-cochain and every element of $\bigwedge ^2 \mathfrak g$ is a $0$-cochain of $\mathfrak g$.
 A linear map $\partial$ which takes $k$-cochains to $(k+1)$-cochains satisfying $\partial^2=0$, is called a {\em coboundary operator}. Since we shall deal with the cases where $k=0$ and $k=1$, we only consider the definition of Chevalley-Eilenberg coboundary operator of these two cases. For $0$-cochain $u\in \bigwedge ^2 \mathfrak g$, we have $\partial u(X)=X.u$, and for $1$-cochain $\delta:\mathfrak g \to \bigwedge ^2 \mathfrak g$ we have:
 \[
 \partial\delta(X,Y)=X.(\delta(Y))-Y.(\delta(X))-\delta([X,Y]) \quad \mbox{for}\quad  X,Y\in \mathfrak g.
 \]
A $k$-cochain is called a {\em $k$-cocycle} if its coboundary is zero.
Now we can proceed to the definition of the Lie bialgebra.

   A {\em Lie bialgebra} $(\mathfrak g, \mathfrak g^*)$ is a Lie algebra with an additional structure, a linear map $\delta:\mathfrak g\to \mathfrak g\otimes \mathfrak g$ such that:
\begin{enumerate}
	\item[$(i)$] The linear map $\delta : \mathfrak{g}\longrightarrow \mathfrak{g}\otimes\mathfrak{g}$ is a $1$-cocycle, i.e.
	$$ad_X^{(2)}(\delta Y)-ad_Y^{(2)}(\delta X)- \delta[X,Y]=0,\quad \forall X,Y\in \mathfrak{g}. $$
	\item[$(ii)$] The dual map $\delta^{t}: \mathfrak{g}^*\otimes{\mathfrak g}^* \longrightarrow \mathfrak g^{*}$ is a Lie bracket on $\mathfrak g^{*}$.
\end{enumerate}	
\noindent	We denote the Lie bracket on $\mathfrak g^*$ by $[\alpha,\beta]_*:=\delta^t(\alpha\otimes\beta)$ for $\alpha, \beta \in \mathfrak g^*$.

{\em Coboundary Lie bialgebra} is a Lie bialgebra defined by a $1$-cocycle $\partial r$ which is the coboundary of an element $r\in \mathfrak g\otimes \mathfrak g$ (for more details, see \cite{Ko}).

\subsubsection{Classical Yang-Baxter equation}
Let $\mathfrak g$ be finite-dimensional Lie algebra and $\mathfrak g^*$ be its dual vector space with respect to a non-degenerate canonical pairing $\left\langle\cdot,\cdot\right\rangle$, so for basis $\{X_i\}$ and dual basis $\{ X^i\}$ of $\mathfrak g$ and $\mathfrak g^*$, respectively, we have:
\[
\langle  X_i,X_j \rangle= \langle  X^i, X^j \rangle=0,
 \quad  \langle  X^i,X_j \rangle=\delta^i_j.
 \]
 To every element $r=r^{ij}X_i\otimes X_j$ of $\mathfrak g \otimes \mathfrak g$, we can associate the linear map $r^{\sharp}:{\mathfrak g}^* \to \mathfrak g$ defined by $r^{\sharp}({X}^i)({X}^j):=r({X}^i,{X}^j)$. Let $\delta_r:=\partial r$, then
 \[
 \delta_ r(X)=ad^{(2)}_{X}r=r^{ij}(ad_X X_i\otimes X_j+X_i\otimes ad_X X_j), \quad X\in \mathfrak g.
 \]
By definition, $\delta r$ is a $1$-cocycle. We denote the bracket on $\mathfrak g^*$ in this case by $[{X}^i,{X}^j]^r$ instead of $[{X}^i,{X}^j]_*$.

 \noindent To every element $r$ of $\mathfrak g \otimes \mathfrak g$ we can also associate a bilinear map $\left\langle r,r\right\rangle ^{\sharp}:{\mathfrak g}^* \times {\mathfrak g}^*\to \mathfrak g$ defined by
 \begin{equation}\label{con-r-matrix}
 \left\langle r,r\right\rangle ^{\sharp}({X}^i,{X}^j)=[r^{\sharp} {X}^i,r^{\sharp} {X}^j]-r^{\sharp}[{X}^i,{X}^j]^r,
 \end{equation}
which can be identified with an element $\langle r,r\rangle \in \wedge^2 \mathfrak g \otimes \mathfrak g$, such that
 \[
\langle r,r\rangle({X}^i,{X}^j,{X}^k):=\langle {X}^k,\left\langle r,r\right\rangle^{\sharp}({X}^i,{X}^j)\rangle.
 \]

\noindent For a skew-symmetric element $r \in \Lambda^2 {\mathfrak g} $, we have
 \begin{equation}\label{r-bracket}
[{X}^i,{X}^j]^r=ad^{*}_{r^{\sharp}{X}^i}{X}^j-ad^{*}_{r^{\sharp}{X}^j}{X}^i,
 \end{equation}
 where $ad^*_{X_i}=-(ad_{X_i})^t$ is the endomorphism of $\mathfrak g^*$  satisfying
 \begin{equation}
 \label{ad1}
 \langle  X^i, ad_{X_k}X_j\rangle =-\langle ad^{*}_{X_k} X^i,X_j \rangle,
 \end{equation}
 which implies
 \begin{equation}
  \label{ad2}
 \langle ad^{*}_{X_k} X^i, X_j\rangle =-\langle ad^{*}_{ X_j} X^i,X_k \rangle.
  \end{equation}

On the other hand, in the case where $r$ is skew-symmetric, we have $\left\langle  r,r \right\rangle =-\frac{1}{2}\lcf r,r\rcf$ where $\lcf \cdot,\cdot\rcf$ is Schouten-Nijenhuis bracket on the Lie algebra $\mathfrak g$ called the {\em algebraic Schouten bracket}.

For the skew-symmetric element $r\in \Lambda^2 {\mathfrak g}$, the condition $\lcf r,r\rcf=0$ is called the {\em classical Yang-Baxter equation} (CYBE). A solution of the CYBE is called an {\em $r$-matrix}. For any $r$-matrix the bracket (\ref{r-bracket}) is a Lie bracket on $\mathfrak g^*$, called the {\em Sklyanin bracket}. Therefore $r$-matrices can be identified with coboundary Lie bialgebras on the Lie algebra $\mathfrak g$ (for more details see, for instance, \cite{Ko}).

A coboundary Lie bialgebra $(\mathfrak g,\mathfrak g^*)$ is called a {\em bi-r-matrix Lie bialgebra}, if the Lie bialgebra $(\mathfrak g^*,\mathfrak g)$ is coboundary and the Sklyanin bracket defined on $\mathfrak g$ is equivalent to the initial Lie bracket on $\mathfrak g$.

\section{Right-invariant Poisson-Nijenhuis structures} \label{Section3}
In this section we define right-invariant $P$-$N$ structures on the Lie group $G$ and their infinitesimal counterpart on the Lie algebra $\mathfrak g$ of $G$. We also consider the concepts of compatibility and equivalence of those structures.

We are using the following notation. If $s$ is a $k$-vector on $\mathfrak g$ then $\rvec{s}$ (resp. $\lvec{s}$) is the right-invariant (resp. left-invariant) $k$-vector field on $G$ given by
\[
\rvec{s}(g) = (T_{\mathfrak e}R_g)(s), \; \; \mbox{ for } g \in G,
\]
(resp. $\lvec{s}(g) = (T_{\mathfrak e}L_g)(s)$, for $g \in G$) where $R_h: G \to G$ and $L_g: G \to G$ are the right and left translation by $h$ and $g$, respectively.

\begin{definition}
	A $P$-$N$ structure $(\Pi, N)$ on a Lie group $G$ is said to be right-invariant if:
	\begin{enumerate}
		\item
		The Poisson structure $\Pi$ is right-invariant, that is, there exists $r \in \Lambda^2 {\mathfrak g}$ such that $\Pi = \rvec{r}$.
		\item
		The Nijenhuis tensor $N$ also is right-invariant, that is, there exists a linear endomorphism $n: {\mathfrak g} \to {\mathfrak g}$ such that
		$N = \rvec{n}$.
	\end{enumerate}
\end{definition}
This last condition means that
\[
N_{|T_gG} = T_{\mathfrak e}R_g \circ n \circ T_gR_{g^{-1}}, \; \; \; \mbox{ for } g \in G.
\]

For right-invariant $P$-$N$ structures, we may prove the two following results which describe the infinitesimal version of such structures.

\begin{proposition}\label{inft-ver-right-1}
	Let  $(N,\Pi)$ be a right-invariant P-N structure on a Lie group $G$ with Lie algebra $\mathfrak g$ and identity element ${\mathfrak e}\in G$. If $r \in \Lambda^2{\mathfrak g}$ is the value of $\Pi$ at ${\mathfrak e}$ and $n$ is the restriction of $N$ to ${\mathfrak g}$, we have:
	\begin{enumerate}
		\item[$(i)$] $r$ is a solution of the classical Yang-Baxter equation on ${\mathfrak g}$.

		\item[$(ii)$]
		$n$ is a Nijenhuis operator on ${\mathfrak g}$, that is,
		\[
		[n, n](X, Y) : = [nX, nY] - n[nX, Y] - n[X, nY] + n^2[X,Y] = 0, \; \; \forall X, Y \in {\mathfrak g}.
		\]
			\item[$(iii)$] $n:\mathfrak g\longrightarrow \mathfrak g$ satisfies the condition
			\[
			n \circ r^{\sharp} = r^{\sharp} \circ n^t.
			\]
		\item[$(iv)$] The concomitant $C(r, n)$ of $r$ and $n$ in ${\mathfrak g}$ is zero, that is,
		\[
		C(r, n)(\alpha, \beta) : =
		\mathcal L^{\mathfrak g}_{r^{\sharp}\alpha}n^{t}\beta - \mathcal L^{\mathfrak g}_{r^{\sharp}\beta}n^{t}\alpha - n^{t}\mathcal L^{\mathfrak g}_{r^{\sharp}\alpha}\beta  + n^{t}\mathcal L^{\mathfrak g}_{r^{\sharp}\beta}\alpha=0\quad \mbox{ for } \alpha, \beta \in \mathfrak g^{*}.
		\]
		Here, ${\mathcal L}^{\mathfrak g}_{X}:=ad^*_{X}:\mathfrak g^* \to \mathfrak g^*$.
	\end{enumerate}
	
Conversely, let ${\mathfrak g}$ be a real Lie algebra of finite dimension, $r \in \Lambda^2{\mathfrak g}$ be a $2$-vector on ${\mathfrak g}$ and $n: {\mathfrak g} \to {\mathfrak g}$ be a linear endomorphism on ${\mathfrak g}$ which satisfy conditions (i), (ii), (iii) and (iv); so-called r-n structure on the Lie algebra $\mathfrak g$. If $G$ is a Lie group with the Lie algebra ${\mathfrak g}$, then the couple $(\rvec{r}, \rvec{n})$ is a right-invariant P-N structure on $G$.
\end{proposition}

\begin{proof}
Using the following results, the proposition would be proved.

 If $\Pi$ is a right-invariant $2$-vector on $G$, and $N$ a right-invariant $(1, 1)$-tensor field on $G$, that is
$\Pi = \rvec{r}$ and $N = \rvec{n}$, then	
	\begin{enumerate}
		\item
		The Schouten-Nijenhuis brackets $[\Pi, \Pi]$ in $G$ and $[r, r]$ in $\frak{g}$ are related by the formula
		\[
		[\Pi, \Pi] = -\rvec{[r, r]}.
		\]
		\item
		The Nijenhuis torsions $[N, N]$ and $[n, n]$ of $N$ and $n$, are related by the formula
		\[
		[N, N](\rvec{X}, \rvec{Y}) = -\rvec{[n, n](X,Y)}, \; \; \mbox{ for } X, Y \in {\frak g}.
		\]
		\item
Using the relations $
\Pi^{\sharp}(\rvec{\alpha}) = \rvec{r^{\sharp}(\alpha)}$ and $N^t\rvec{\alpha} = \rvec{n^t\alpha}$, for $\alpha \in \frak{g}^*$, we have
		\[
		(N \circ \Pi^{\sharp} - \Pi^{\sharp} \circ N^t)(\rvec{\alpha}) = \rvec{(n \circ r^{\sharp} - r^{\sharp} \circ n^t)(\alpha)}.
		\]
		\item
		If $C(\Pi, N)$ (resp. $C(r, n)$) is the concomitant of $\Pi$ and $N$ in $G$ (resp. $r$ and $n$ in $\frak{g}$) then
		\[
		C(\Pi, N)(\rvec{\alpha}, \rvec{\beta}) = -\rvec{C(r, n)(\alpha, \beta)}, \mbox{ for } \alpha, \beta \in \frak{g}^*.
		\]
To show the above relation, by using (\ref{Def-Concomi}) and the fact that $
		{\mathcal L}_{\rvec{X}}\rvec{\alpha} = -\rvec{\mathcal{L}^{\frak{g}}_{X}\alpha}$ \cite{AZ}, for $X \in {\frak g} $ and $\alpha \in {\frak g}^*$, we obtain
		\[
		C(\Pi, N)(\rvec{\alpha}, \rvec{\beta}) = - \rvec{C(r, n)(\alpha, \beta)} + d \langle \rvec{\alpha}, N \Pi^{\sharp}\rvec{\beta} \rangle
		- N^t d\langle \rvec{\alpha}, \Pi^{\sharp}\rvec{\beta} \rangle.
		\]
	Since
		\[
		\langle \rvec{\alpha}, N \Pi^{\sharp}\rvec{\beta} \rangle = \langle \alpha, n(r^{\sharp}\beta) \rangle, \; \; \; \langle \rvec{\alpha}, \Pi^{\sharp}\rvec{\beta} \rangle = \langle \alpha, r^{\sharp}\beta \rangle,
		\]
		we conclude that $C(\Pi, N)(\rvec{\alpha}, \rvec{\beta}) = -\rvec{C(r, n)(\alpha, \beta)}$.
\end{enumerate}	\end{proof}	
We remark that solutions of the CYBE on $\mathfrak g$ are in one-to-one correspondence with right-invariant (resp. left-invariant) Poisson structures on $G$. In fact, $r \in \Lambda^2 {\mathfrak g}$ is a solution of the CYBE on $\mathfrak g$ if and only if $\rvec{r}$ (resp. $\lvec{r}$) is a right-invariant (resp. left-invariant) Poisson structure on $G$.

All the Poisson structures in the hierarchy associated with a right-invariant $P$-$N$ structure are also right-invariant. 

\begin{corollary}\label{hierarchy-right-inv}
Let $G$ be a Lie group with Lie algebra ${\mathfrak g}$, $r$ be a solution of the CYBE on ${\mathfrak g}$ and $n: {\mathfrak g} \to {\mathfrak g}$ be a linear endomorphism such that the couple $(\rvec{r}, \rvec{n})$ is a right-invariant P-N structure on $G$.

\begin{enumerate}
\item[$(i)$]
If $r_k$, with $k \in \mathbb{N}$, is the $2$-vector on $\mathfrak{g}$ which is characterized by the condition $r_k^{\sharp} = n^k \circ r^{\sharp}$, then $r_k$ is a solution of the CYBE.
																		\item[$(ii)$]
If $\Pi_k$, with $k \in \mathbb{N}$, is the Poisson structure on $G$ which is characterized by the condition $\Pi_k^{\sharp} = {\rvec{n}}^k \circ \rvec{r}^{\sharp}$, then $\Pi_k$ is right-invariant. In fact, $\Pi_k = \rvec{r_k}$.
	\end{enumerate}
\end{corollary}

\begin{proof}  Using ${\rvec{n}}^k = \rvec{n^k}$, it follows that $\Pi_k^{\sharp} = \rvec{n^k} \circ \rvec{r}^{\sharp}$; thus $\Pi_k = \rvec{r_k}$. Since $(\rvec{r}, {\rvec{n}}^k)$ is a $P$-$N$ structure, $\rvec{r_k}$ is a right-invariant Poisson structure, which implies that $r_k$ is a solution of the CYBE on ${\mathfrak g}$.
		
\end{proof}

 \subsection{Compatibility of right-invariant Poisson-Nijenhuis structures} \label{compatible}
 Two $P$-$N$ structures $(\Pi', N')$ and $(\Pi, N)$ on a Lie group $G$ are said to be {\em compatible} if the couple $(\Pi'+\Pi, N'+N)$ is a $P$-$N$ structure on $G$.

 Note that the compatibility condition of two Nijejhuis structures has been considered in \cite{Ni} as {\em Nijenhuis concomitant}.
The Nijenhuis concomitant of two $(1,1)$-tensor fields  $n'$ and $n $ on the Lie algebra $\mathfrak g$ with vanishing Nijenhuis torsion is defined by
\begin{equation}\label{Ni-con}
\begin{array}{rcl}
[n',n](X_i,X_j)& =& \left[ n'X_i,nX_j\right]  - n'[nX_i, X_j] - n'[X_i,nX_j] + n'\circ n[X_i,X_j]\\[4pt]
&+&[nX_i,n'X_j] - n[n'X_i, X_j] - n[X_i,n'X_j] + n\circ n'[X_i,X_j],
\end{array}
\end{equation}
for every two elements of basis $\{X_i\}$ of $\mathfrak g$. Two Nijenhuis structures are compatible if the Nijenhuis concomitant of them vanishes \cite{Ni}. Recall that two $r$-matrices are compatible if their sum is still an $r$-matrix.

In the case of the right-invariant $P$-$N$ structures, the compatibility of structures $(\Pi',N')$ and $(\Pi,N)$ reduces to the compatibility of their infinitesimal version.

 If $(r',n')$ and $(r,n)$ be the infinitesimal versions of the mentioned $P$-$N$ structures, then $(r',n')$ and $(r,n)$ are {\em compatible} if the couple $(r'+r,n'+n)$ is an $r$-$n$ structure on the Lie algebra $\mathfrak g$ of $G$. From the Proposition \ref{inft-ver-right-1}, $(r'+r,n'+n)$ would be the infinitesimal version of the right-invariant $P$-$N$ structure $(\rvec {r'+r},\rvec{n'+n})$ on the Lie group $G$.
 \begin{remark}\label{Rem}
 Let $(r',n')$ and $(r,n)$ be two $r$-$n$ structures such that $r'$ and $r$ are compatible $r$-matrices and $n'$ and $n$ are compatible Nijenhuis structures. It is easy to see that if $(r',n)$ and $(r,n')$ are also $r$-$n$ structures, then $(r',n')$ and $(r,n)$ are compatible $r$-$n$ structures.
 
 \end{remark}

  \subsection{Equivalence classes of right-invariant Poisson-Nijenhuis structures}\label{equivalence.}
We say two right-invariant Poisson structures  $\Pi'=\rvec{ r'}$ and $\Pi=\rvec{ r}$ on the Lie group $G$ are {\em equivalent} if the solutions $r',r\in \wedge ^2 \mathfrak g$ of the CYBE are equivalent, i.e. there exist a Lie algebra automorphism $\mathcal A$ such that the diagram $(a)$ below commutes. In the same way, we say two right-invariant Nijenhuis structures $N'=\rvec{ n'}$ and $N=\rvec{ n}$ on the Lie group $G$ are {\em equivalent} if two linear map $n'$ and $n$ are equivalent, i.e. there exist a Lie algebra automorphism $\mathcal A$ such that diagram $(b)$ below commutes. In these cases we will write $r'\sim r$ and $n'\sim n$ ($r'\sim_\mathcal{A}r$ and $n'\sim_\mathcal{A} n$ if we want to indicate the automorphism $\mathcal{A}$).
  \[
  \begin{tikzpicture}
  \centering
  \matrix (m) [matrix of math nodes,row sep=3em,column sep=4em,minimum width=2em]
  {
  	\mathfrak g^* &  \mathfrak g \\
  	\mathfrak g^*& \mathfrak g  \\};
  \path[-stealth]
  (m-1-1) edge node [left] {$(a):\quad  \mathcal A^{-t}$} (m-2-1)
  edge  node [above] {$r'^{\sharp}$} (m-1-2)
  (m-2-1.east|-m-2-2) edge node [below] {$r^{\sharp}$}
  node [above] {$$} (m-2-2)
  (m-1-2) edge node [right] {$\mathcal{A}$} (m-2-2);
  \end{tikzpicture}\qquad%
  \]
  \[
  \begin{tikzpicture}
  \centering
  \matrix (m) [matrix of math nodes,row sep=3em,column sep=4em,minimum width=2em]
  {
  	\mathfrak g &  \mathfrak g \\
  	\mathfrak g& \mathfrak g  \\};
  \path[-stealth]
  (m-1-1) edge node [left] {$(b):\quad \mathcal A$} (m-2-1)
  edge  node [above] {$n'$} (m-1-2)
  (m-2-1.east|-m-2-2) edge node [below] {$n$}
  node [above] {$$} (m-2-2)
  (m-1-2) edge node [right] {$\mathcal{A}$} (m-2-2);
  \end{tikzpicture}
  \]

Therefore we can talk about equivalence classes of right-invariant $P$-$N$ structures as well. Two right-invariant $P$-$N$ structures $(P',N')$ and $(P,N)$ on the Lie group $G$ are {\em equivalent} if two corresponding $r$-$n$ structures $(r',n')$ and $(r,n)$ on the Lie algebra $\mathfrak g$ are equivalent.

\begin{definition}\label{equivalnt}
	Two r-n structures $(r',n')$ and $(r,n)$ are equivalent if there exist a Lie algebra automorphism $\mathcal A$ such that the following diagram commutes.
	\[
	\begin{tikzpicture}
	\matrix (m) [matrix of math nodes, row sep=3em, column sep=4em]
	{  & \mathfrak g^*  & \mathfrak g  & \mathfrak g &  \\
		& \mathfrak g^*  & \mathfrak g & \mathfrak g &  \\ };
	{ [start chain] 
		\chainin (m-1-2);
		{ [start branch=A] \chainin (m-2-2)
			[join={node[left,labeled] {\mathcal A ^{-t}}}];}
		\chainin (m-1-3) [join={node[above,labeled] {r'^{\sharp}}}];
		{ [start branch=B] \chainin (m-2-3)
			[join={node[right,labeled] {\mathcal A}}];}
		\chainin (m-1-4) [join={node[above,labeled] {n'}}];
		{ [start branch=C] \chainin (m-2-4)
			[join={node[right,labeled] {\mathcal A}}];}
	}
	{ [start chain] 
		\chainin (m-2-2);
		\chainin (m-2-3) [join={node[below,labeled] {r^{\sharp}}}];
		\chainin (m-2-4) [join={node[below,labeled] {n}}];}
		\end{tikzpicture}
		\]

 In this case we will write $(r',n')\sim (r,n)$ ($(r',n')\sim_\mathcal{A}(r,n)$ if we want to indicate $\mathcal{A}$).
	\end{definition}

\section{Compatible $r$-matrices and $r$-$n$ structures}\label{Section4}
As we saw in Corollary \ref{hierarchy-right-inv}, having an $r$-$n$ structure $(r,n)$ on the Lie algebra $\mathfrak g$ we get a new solution of the CYBE $(r')^{\sharp}=n \circ r^{\sharp}$. In this section we will show that two r-matrices, under a certain condition, are related to $r$-$n$ structures on $\mathfrak g$.

Throughout this section we denote the basis $\{X_i\}$ and the dual basis $\{ X^j\}$ for Lie algebras $\mathfrak g$ and $\mathfrak g^*$, respectively.

Let $r\in \Lambda^2 {\mathfrak g}$ be an invertible solution of the CYBE, i.e. $r^{\sharp}:{\mathfrak g}^* \to \mathfrak g$ is an isomorphism map, then $(r^{-1})^{\sharp}:{\mathfrak g} \to \mathfrak g^*$, which corresponds to an element $r^{-1} \in  \Lambda^2 g^*$ defined by
	\[
	(r^{-1})^{\sharp}({X}_i)({X}_j):=(r^{-1})({X}_i,{X}_j),\quad X_i \in \mathfrak g,
	\]
	is a solution of the CYBE $\lcf r^{-1},r^{-1}\rcf^{r}=0$, where $\lcf \cdot,\cdot \rcf^{r}$ is the algebraic Schouten-Nijenhuis bracket on $\mathfrak g^*$. The relation between two maps $r^{\sharp}$ and $(r^{-1}) ^{\sharp}$ is given by
 \[
 \langle (r^{-1})^{\sharp}(X_i) , r^{\sharp}( X^j) \rangle =-\langle  X^j ,X_i\rangle.
 \]
\noindent We set $r=r^{ij}X_i\otimes X_j$ and $r^{-1}=(r^{-1})_{ij} X^i\otimes  X^j$; it is obvious that $(r^{\sharp})^{ij}=r^{ij}$ and $(r^{-1})^{\sharp}_{ij}=(r^{-1})_{ij}$, where $(r^{\sharp})(X^i)=(r^{\sharp})^{ij}X_j$ and $(r^{-1})^{\sharp}( X_i)=(r^{-1})^{\sharp}_{ij} X^j$ for all $ X^i \in \mathfrak g^*$ and $X_i \in \mathfrak g$.

If there is not risk of confusion, we will use the same notation $r$ for the 2-vector $r$ and the linear map $r^{\sharp}$.
\begin{proposition}
	Let $r\in \Lambda^2 {\mathfrak g}$ be an invertible solution of the CYBE and $[\cdot,\cdot]^{r^{-1}}$ be the Sklynin bracket on $\mathfrak g$ defined by $r^{-1}$ on the dual Lie algebra $(\mathfrak g^*,[\cdot,\cdot]^r)$, then $[\cdot,\cdot]^{r^{-1}}=[\cdot,\cdot]$.
\end{proposition}
\begin{proof}
	Since $r$ is an $r$-matrix, we have $
	r[ X^i,X^j]^r=[r X^i,r X^j]$. On the other hand $r^{-1}$ is also an r-matrix which implies
 \begin{equation}\label{relation2}
 r^{-1}\left[ X_i,X_j\right] ^{r^{-1}}=\left[r^{-1}X_i ,r^{-1} X_j\right] ^r=r^{-1}[X_i,X_j],
 \end{equation}
and the proof is complete.

\end{proof}
\noindent Therefore, the Lie bialgebras identified by invertible $r$-matrices are bi-$r$-matrix Lie bialgebras. We consider the following notation for the Sklaynin brackets on $\mathfrak g^*$ and $\mathfrak g$ defined by $r$-matrices $r$ and $r^{-1}$:
\[
[ X^i,  X^j]^r:=[ X^i,  X^j]_*={ \tilde f}^{ij}_{k} X^k, \quad \left[ X_i,X_j\right] ^{r^{-1}}:=[X_i,X_j]={f_{ij}}^{k} X_k,
\]

\noindent where ${\tilde  f}^{ij}$ and $f_{ij}$ are structure constants of $\mathfrak g^*$ and $\mathfrak g$, respectively.

\begin{proposition}\label{P-N}
	Let $2$-vector $r\in \Lambda^2 {\mathfrak g}$ be an invertible solution of the CYBE and let $r'\in \Lambda^2 {\mathfrak g}$  be another solution of the CYBE (not necessary invertible) which is compatible with $r$, then the endomorphism $r'\circ r^{-1}: \mathfrak g \to \mathfrak g$,
	\begin{enumerate}
		\item[$(i)$] is a Nijenhuis operator on $\mathfrak g$ with respect to the Lie bracket $[\cdot,\cdot]^{r^{-1}}=[\cdot,\cdot]$.

		\item[$(ii)$] is compatible with $r$ that means $( r,r'\circ r^{-1})$ is an r-n structure on $\mathfrak g$.
	\end{enumerate}	
\end{proposition}

\begin{proof}
\noindent $(i)$
Since $r$ and $r'$ are compatible, we have
\[
(r+r')[X^i,  X^j] ^{r+r^{'}}=[ (r+r')( X^i), (r+r')( X^j)].
\]
By using (\ref{r-bracket}) and the fact that $r$ and $r'$ are  solutions of the CYBE, we get
\begin{equation}\label{first-relation}
r[  X^i, X^j] ^{r^{'}}+r^{'}[  X^i,  X^j] ^{r}=[ rX^i,r^{'} X^j] +[r^{'} X^i,r X^j].
\end{equation}
Now consider
\[
(r'\circ r^{-1})^2[ X_i,X_j] =(r'\circ r^{-1})r'[ r^{-1}X_i,r^{-1}X_j] ^r.
\]

\noindent By applying (\ref{first-relation}) in the above relation, we get
\[
\begin{array}{rclcrcl}
(r'\circ r^{-1})^2[X_i,X_j]&=& (r'\circ r^{-1}) \left[ r (r^{-1}X_i),r'(r^{-1}X_j)\right]+(r'\circ r^{-1})\left[ r'(r^{-1}X_i),r(r^{-1}X_j)\right] \\[6pt]
 &&-(r'\circ r^{-1})r\left[ r^{-1}X_i,r^{-1}X_j\right] ^{r'},
\end{array}
\]
which implies
\[
(r'\circ r^{-1})^2[X_i,X_j]= r'\circ r^{-1} \left[ X_i,(r'\circ r^{-1})X_j\right] +r'\circ r^{-1}\left[ (r'\circ r^{-1})X_i,X_j\right] -\left[ (r'\circ r^{-1})X_i,(r'\circ r^{-1})X_j\right],
 \]
and the proof is complete.

\noindent $(ii)$
It is easy to check that the first compatibility condition holds, since $r$ and $r'$ are skew-symmetry we have $(r^{-1})^{t}=-r^{-1}$ and $(r')^{t}=-r'$.

To show the second compatibility condition (the concomitant), observe that since $r$ is an isomorphism, for every element $X_i\in \mathfrak g$ there exist an element $ X^k \in \mathfrak g^*$ such that $X_i:=r{X^k}$.

 We will prove
 \begin{equation}\label{concomitant}
\begin{array}{rcl}
\langle  \mathcal L^{\mathfrak g}_{r X^i} (r^{-1}\circ r') X^j, r X^k\rangle &=& \langle  \mathcal L^{\mathfrak g}_{r X^j}(r^{-1}\circ r') X^i, r X^k\rangle+\langle (r^{-1}\circ r')\mathcal L^{\mathfrak g}_{r X^i} X^j, r X^k\rangle\\[7pt]
      &-&\langle (r^{-1}\circ r')\mathcal L^{\mathfrak g}_{r X^j} X^i, r X^k\rangle.\\[6pt]
    \end{array}
    \end{equation}
Using relations
\[
\langle \mathcal L^{\mathfrak g}_{X_i} X^j,X_k\rangle =-\langle  X^j,\mathcal L^{\mathfrak g}_{X_i}X_k\rangle,\quad
\langle (r^{-1}\circ r') X^i,X_j\rangle =\langle X^i,(r'\circ r^{-1})X_j\rangle,
\]
and (\ref{relation2}), we get
 \[
\begin{array}{rcl}
\langle  \mathcal L^{\mathfrak g}_{r X^i} (r^{-1}\circ r') X^j, r X^k\rangle
  &=& -\langle  X^j , r'[ X^i,  X^k]^r\rangle.  \\[5pt]
  \end{array}
  \]
  Using (\ref{first-relation}), the above relation gives
  \begin{equation}\label{1}
  =-\langle  X^j,[ r X^i,r' X^k]\rangle
 -\langle  X^j,[r' X^i,r X^k]\rangle
  +\langle  X^j, r[  X^i, X^k] ^{r'} \rangle.
\end{equation}

 \noindent Similarly, we see that
 \begin{equation}\label{2}
\langle (r^{-1}\circ r')\mathcal L^{\mathfrak g}_{r X^i} X^j, r X^k\rangle=-\langle  X^j,[ r X^i,r' X^k]\rangle.
 \end{equation}
Now, by using (\ref{r-bracket}), (\ref{ad1}), (\ref{ad2}) and $\langle X^i, r X^j\rangle=-\langle X^j, r X^i\rangle$, we have:
\[
\begin{array}{rcl}
&& -\langle  X^j,[ r' X^i,r X^k]\rangle
  +\langle  X^j, r[  X^i, X^k] ^{r'} \rangle \\[8pt]
  &&=-\langle ad^{*}_{r X^k} X^j, r' X^i \rangle+\langle ad^{*}_{r X^j}X^k, r' X^i \rangle+\langle ad^{*}_{r' X^k} X^i, r X^j \rangle\\[8pt]
  &&=\langle [ X^j,X^k] ^r,r' X^i\rangle -\langle  X^i, [r' X^k,r X^j] \rangle.
  \end{array}
\]
On the other hand, in the same way we observe that
\begin{equation}\label{3}
\begin{array}{rcl}
\langle  \mathcal L^{\mathfrak g}_{r X^j}(r^{-1}\circ r') X^i, r X^k\rangle&=&\langle  X^i , r'[  X^j, X^k] ^r\rangle, \\[8pt]
\langle (r^{-1}\circ r')\mathcal L^{\mathfrak g}_{r X^j} X^i, r X^k\rangle&=& \langle  X^i,[ r' X^k,r X^j] \rangle,
\end{array}
\end{equation}
so, from (\ref{1}), (\ref{2}) and (\ref{3}), we get (\ref{concomitant}) and the proof is complete.

\end{proof}

\begin{corollary}\label{corollary}
	Let $r$, $r'$ and $r''$ be skew-symmetric solutions of the CYBE and let $r$ be an invertible solution which is compatible with $r'$ and $r''$, then  $n':=r'\circ r^{-1}$ and $n'':=r''\circ r^{-1}$ are compatible if and only if
	$r'$ and $r''$ are compatible.
	
\end{corollary}
\begin{proof}
From Proposition \ref{P-N}, $(r,n')$ and $(r,n'')$ are r-n structures. If $n'$ and $n''$ are compatible, then from Remark \ref{Rem}, $(r,n')$ and $(r,n'')$ are compatible $r$-$n$ structures which means that  $(r,n'+n'')$ is an $r$-$n$ structure and therefore $(n'+n'')\circ r$ is an $r$-matrix, so
\[
(n'+n'')\circ r=n'\circ r+n''\circ r=r'+r'',
\]
 simply implies $r'$ and $r''$ are compatible. One proves the converse in a similar way.
 
\end{proof}	
\rm
Note that, if $(r,n')$ and $(r,n'')$ are compatible $r$-$n$ structures, then $n' \circ r$ and $n''\circ r$ are compatible solutions of the CYBE.
\begin{remark}\label{Rem2}
	If $r$ is an invertible solution of the CYBE, then there is a one-to-one correspondence between $r$-matrices $r^l$ and Nijenhuis operators $n^l$ such that $r^l:=n^l\circ r$; or equivalently, there is a one-to-one correspondence between $r$-$n$ structures $(r,n^l)$ and compatible $r$-matrices $(r,r^l)$.
	
\end{remark}

\begin{corollary}\label{Theorem2}
	Let $r$ be an invertible solution of the CYBE on the Lie algebra $\mathfrak g$. Then, r-n structures $(r,n)$ and $(r',n')$ are equivalent if and only if the couple of compatible r-matrices $(r,r_0)$ and $(r',r'_0)$ are equivalent, where $r_0=n\circ r$ and $r'_0=n'\circ r'$.	
\end{corollary}

\begin{proof}
According to Definition \ref{equivalnt}, $(r,n)\sim(r',n')$ if for some automorphism  $\mathcal A$ of the Lie algebra $\mathcal A\circ r=r'\circ {\mathcal A}^{-t}$and $\mathcal A\circ n=n'\circ \mathcal A$, which implies $
\mathcal A \circ n \circ r=n' \circ r' \circ \mathcal A^{-t}$, so $r_0=n\circ r \sim_\mathcal{A} r'_0=n' \circ r'$ and $(r,r_0) \sim_\mathcal{A}  (r',r'_0)$. Conversely, if $(r,r_0) \sim_\mathcal{A}  (r',r'_0)$, then $\mathcal A\circ r=r'\circ {\mathcal A}^{-t}$ and $\mathcal A\circ r_0=r'_0\circ {\mathcal A}^{-t}$. Hence,
$${\mathcal A}n_0{\mathcal A}^{-1}={\mathcal A}r_0\circ r^{-1}{\mathcal A}^{-1}=\mathcal A r_0{\mathcal A}^t\circ{\mathcal A}^{-t}r^{-1}{\mathcal A}^{-1}=r'_0\circ(r')^{-1}=n'$$
and $(r,n)\sim_{\mathcal A}(r',n')$.

\end{proof}

 \section{classification procedure}\label{method}
 In this section we describe the classification procedure of $r$-$n$ structures on $\mathfrak g$, equivalently right-invariant $P$-$N$ structures on $G$. For this purpose we rewrite the four conditions of Proposition \ref{inft-ver-right-1} in terms of coordinates.

First, we write the structural constants $f^{k}_{ij}$ of the Lie algebra $\mathfrak g$, in terms of adjoint representation ${\mathcal {X}_{i}}$, and antisymmetric matrices ${\mathcal Y^i}$, as
 \begin{equation}\label{Matrix}
 f^{k}_{ij}=-(\mathcal{Y}^{k})_{ij},\quad f^{k}_{ij}=-(\mathcal {X}_{i})_{j}^{~k}.
 \end{equation}

\noindent{\bf Condition (i)} Consider the tensor notation of the CYBE, (see \cite{Ko})
\[
\lcf r,r \rcf =[r_{12},r_{13}]+[r_{12},r_{23}]+[r_{13},r_{23}]=0,
\]
 where $
r_{12}=r^{ij}X_i\otimes X_j\otimes 1$, $r_{13}=r^{ij}X_i\otimes 1\otimes X_j$ and $r_{23}=r^{ij} 1\otimes X_i\otimes X_j$. Applying the above condition in the base elements, we get (for more details see \cite{Ko})
\[
\lcf r,r \rcf(X^i,X^j,X^k)=\langle X^i, [r^tX^j,r^tX^k]\rangle+\langle X^j, [rX^i,r^tX^k]\rangle+\langle X^k, [rX^i,rX^j]\rangle=0,
\]
which implies
\[
 r^{il}r^{m j}f^{k}_{lm}+r^{il}r^{km}f^{j}_{lm}+r^{kl}r^{j m}f^{i}_{lm}=0.
\]

\noindent Finally, by (\ref{Matrix}), it can be rewritten in the matrix form
\begin{equation}\label{Poisson}
  r{\mathcal Y}^ {i}r-r{\mathcal X}_{l}r^{il}-r^{il}{\mathcal X}_{l}^{t}r=0,\quad i:=1,...,dim \mathfrak g.
\end{equation}
  \noindent{\bf Condition (ii)}  We rewrite Nijenhuis torsion for two base elements ${X_i}$ and ${X_j}$ in $\mathfrak g$ and we get
 \[
 f^{k}_{ij}n^{l}_{~k}n^{m}_{~l}-n^{k}_{~i}n^{m}_{~l}f^{l}_{kj}-n^{k}_{~j}f^{l}_{ik}n^{m}_{~l}+n^{k}_{~i}n^{l}_{~j}f^{m}_{kl}=0.
 \]
 By using (\ref{Matrix}) it can be rewritten in the matrix form
 \begin{equation}\label{Nijenhuis}
 {\mathcal X} _in^{t}n^{t}+n^l
 _{~i}n^{t}{\mathcal X}_{l}-n^{t}{\mathcal X}_{i}n^{t}-n^{l}_{~i}{\mathcal X}_{l}n^{t}=0,\quad i:=1,...,dim \mathfrak g.
 \end{equation}
  \noindent{\bf Condition (iii)} For every element $X^{i}$ in $\mathfrak g^*$, $n\circ r^{\sharp}(X^i)=r^{\sharp}\circ n^t(X^i)$ implies
  \begin{equation}\label{Con1}
  (nr)^i_{~j}=(rn^t)^i_{~j}, \quad i,j=1,...,dim(\mathfrak g),
  \end{equation}
  where $n$ and $r$ are the corresponding matrices to the linear operator $n$ and the map $r^{\sharp}$ and $nX_i=n_i^jX_j$ for $n_i^j\in \mathbb R$.

 \noindent{\bf Condition (iv)} By applying $ X^{i}$ and $X^{j}$ in the concomitant, we get
  \[
r^{l i}n^{m}_{~k}f^{j}_{lm}+r^{jl}n^{m}_{~k}f^{i}_{l m}+r^{jl}n^{i}_{~m}f^{m}_{kl}+r^{l i}n^{j}_{~m}f^{m}_{kl}=0,
 \]

\noindent and then, using (\ref{Matrix}), we get the matrix relation
 \begin{equation}\label{Con2}
 r{\mathcal X}_{j}n^{j}_{~i}+{\mathcal X}^{t}_{j}n^{j}_{~i}r-r{\mathcal X}_{i}n^{t}-n{\mathcal X}^{t}_{i}r=0,\quad i:=1,...,dim \mathfrak g.
 \end{equation}
Given a Lie algebra $\mathfrak g$, by applying matrices ${\mathcal X}_i$ and ${\mathcal Y}^i$ using (\ref{Matrix}) in four equations (\ref{Poisson}), (\ref{Con1}), (\ref{Nijenhuis}), (\ref{Con2}) and solving them by help of mathematical softwares, one can find all $r$-$n$ structures on $\mathfrak g$ and so, all right-invariant $P$-$N$ structures on the Lie group $G$.
 \section{Classification of all $r$-matrices and all $r$-$n$ structures with invertible $r$ \\ on four-dimensional symplectic real Lie algebras}\label{cllasification}

In this section we exemplify results of Sections \ref{Section4} and \ref{method} in classifying, up to an equivalence, all $r$-matrices, and all $r$-$n$ structures with invertible $r$-matrices on four-dimensional symplectic real Lie algebras.

Recall that a symplectic structure on a $2n$-dimensional Lie algebra $\mathfrak g$ is defined as a closed $2$-form $\omega$ which has maximal rank, that is $\omega ^n$ is a volume form on the corresponding Lie group. The list of four-dimensional real Lie algebras with symplectic structure is given in \cite{GOv}.
Four dimensional real symplectic real Lie bialgebras have been classified in \cite{Abedi}.

The strategy is as follows. First, using (\ref{Poisson}) we find all $r$-matrices on four-dimensional symplectic real Lie algebra $\mathfrak g$ and classify them up to equivalence 
\[
r\sim r'\quad \Leftrightarrow \quad\exists \mathcal{A}\in Aut(\mathfrak{g}) \quad \mathcal{A}r\mathcal{A}^t=r'.
\]

\noindent Then, we take a representative of invertible one, and find all $r$-$n$ structures $(r,n)$ on $\mathfrak g$ by solving equations  (\ref{Nijenhuis}), (\ref{Con1}) and (\ref{Con2}) which give nine equations on four-dimensions. We did all computations using Maple. Finally, we classify all obtained pairs $(r,n)$ up to equivalence 
	\begin{equation}\label{e1}(r,n')\sim_0 (r,n) \quad \Leftrightarrow \quad\exists \mathcal{A}\in Aut(\mathfrak{g}) \quad\mathcal{A}r\mathcal{A}^t=r \ \&\ \mathcal{A}n\mathcal{A}^{-1}=n'\, ,
		\end{equation}
	where $\sim _{0}$ indicates the equivalence for $r$-$n$ structures with the same $r$. 

	\begin{proposition} \label{Pro}
		If $\{r_\alpha\}_\Lambda$ is a set of all representatives of the equivalence relation $r\sim r'$ and $\{(r_\alpha,n_\beta)\}_{(\alpha,\beta)\in\Gamma}$is a set of all representatives of the equivalence relations
		$(r_\alpha,n)\sim_0(r_\alpha,n')$, then $\{(r_\alpha,n_\beta)\}_{(\alpha,\beta)\in\Gamma}$ is a set of representatives of the equivalence relation $(r,n)\sim(r',n')$ (c.f. Definition \ref{equivalnt}).
	\end{proposition}
	\begin{proof}
		Consider an r-n structure $(r,n)$. There exist $r_\alpha$ and $\mathcal{A}$ such that $r_\alpha=\mathcal{A}r\mathcal{A}t$. Take the representative $(r_\alpha,n_\beta)$ of $(r_\alpha, \mathcal{A}n\mathcal{A}^{-1})$. Then, it is easy to see that $(r_\alpha,n_\beta)$ represents the class of $(r,n)$ 
		under $\sim$. 
		
		Moreover, different elements of $\{(r_\alpha,n_\beta)\}_{(\alpha,\beta)\in\Gamma}$ represent different elements of $\sim$. Indeed, if $(r_\alpha,n_\beta)\sim(r_{\alpha'},n_{\beta'})$, then there is $\mathcal{A}$ such that
		$r_\alpha=\mathcal{A}r_{\alpha'}\mathcal{A}^t$ and $n_\beta=\mathcal{A}n_{\beta'}\mathcal{A}^{-1}$.
		But then, by definition, $\alpha=\alpha'$ and hence $(r_\alpha,n_\beta)\sim_0(r_\alpha,n_{\beta'})$, thus also $\beta=\beta'$.
		
	\end{proof}
\noindent In the following example we clarify the above procedure by describing the details for a Lie algebra.

{\bf Example}.
For the Lie algebra $A_{4,1}$ \footnote{We use the notations of the four-dimensional real Lie algebras denoted in \cite{Abedi}, (see also \cite{Christ}).} with non-zero commutators $[X_2,X_4]=X_1$ and $[X_3,X_4]=X_2$, from (\ref{Matrix}) we have the following matrices $\mathcal{X}_i$ and $\mathcal{Y}^i$

\[
   {\footnotesize	\mathcal{X}_1=\left(\begin{array}{cccc}
 	0& 0& 0 & 0\\ 0& 0& 0 & 0\\ 0& 0& 0 & 0\\ 0& 0& 0 & 0
 	\end{array} \right),\quad 	\mathcal{X}_2=\left(\begin{array}{cccc}
 0& 0& 0 & 0\\ 0& 0& 0 & 0 \\ 0& 0& 0 & 0\\ -1& 0& 0 & 0
 	\end{array} \right),\quad 	\mathcal{X}_3=\left(\begin{array}{cccc}
 	0& 0& 0 & 0\\ 0& 0& 0 & 0 \\ 0& 0& 0 & 0\\ 0& -1& 0 & 0
 	\end{array} \right),\quad 	\mathcal{X}_4=\left(\begin{array}{cccc}
 	0& 0& 0 & 0\\ 1& 0& 0 & 0 \\ 0& 1& 0 & 0\\ 0& 0& 0 & 0
 	\end{array}\right)},
 	\]
 	\[
   {\footnotesize	\mathcal{Y}^1=\left(\begin{array}{cccc}
 	0& 0& 0 & 0\\ 0& 0& 0 & -1\\ 0& 0& 0 & 0\\ 0& 1& 0 & 0
 	\end{array} \right),\quad
 		\mathcal{Y}^2=\left(\begin{array}{cccc}
 	0& 0& 0 & 0\\ 0& 0& 0 & 0 \\ 0& 0& 0 & -1\\ 0& 0& 1 & 0
 	\end{array} \right),\quad 	
 	\mathcal{Y}^3=\left(\begin{array}{cccc}
 	0& 0& 0 & 0\\ 0& 0& 0 & 0 \\ 0& 0& 0 & 0\\ 0& 0& 0 & 0
 	\end{array} \right),\quad
 		\mathcal{Y}^4=\left(\begin{array}{cccc}
 	0& 0& 0 & 0\\ 0& 0& 0 & 0 \\ 0& 0& 0 & 0\\ 0& 0& 0 & 0
 	\end{array} \right).}
 	\]
 	
\subsection {\bf  Classification of $r$-matrices}	 \label{class-r-matrices}
 
Take a generic 2-vector $r=\sum_{i,j=1}^{4}r^{ij}X_i\wedge X_j$. Inserting matrices $\mathcal{X}_i$ and $\mathcal{Y}^i$, and the matrix form of $r$ in (\ref{Poisson}), one can obtain all possible $r$-matrices on the Lie algebra $A_{4,1}$ as
\[
r=r^{12}X_1\wedge X_2 +r^{13}X_1\wedge X_3+r^{14}X_1\wedge X_4+r^{23}X_2\wedge X_3.
\]
\noindent We use the following automorphism group element (classified in \cite{Christ}, see also \cite{ReSe1}) of this algebra
 \[
 \footnotesize{
 	\mathcal A=\left(\begin{array}{cccc}
 	a_{11} a_{16}^2 & a_7 a_{16} & a_3 & a_4\\ 0 & a_{11} a_{16}& a_7 & a_8\\0 & 0 & a_{11} & a_{12}\\ 0 & 0 & 0 & a_{16}
 	\end{array} \right).}
 \]
\noindent We list the automorphism groups of four-dimensional symplectic real Lie algebra in the table 4 in Appendix for self containing of the paper.
\noindent Inserting the above $\mathcal A$ in the relation $\mathcal A\circ r-r'\circ \mathcal A^{-t}=0$, where $r'=r'^{12}X_1\wedge X_2 +r'^{13}X_1\wedge X_3+r'^{14}X_1\wedge X_4+r'^{23}X_2\wedge X_3$, we get
\[
r'^{14}=a_{11}a_{16}^3r^{14}, \quad 
r'^{13}=(a_{11}r^{13}+a_{12}r^{14})a_{11}a_{16}^2+a_7a_{11}a_{16}r^{23}, 
\]
\[
r'^{23}=a_{11}^2a_{16}r^{23},\quad r'^{12}=(a_{11}a_{16}r^{12}+a_7r^{13}+a_8r^{14})a_{11}a_{16}^2+a_7^2a_{16}r^{23}-a_3a_{11}a_{16}r^{23}.
\]
Looking at the equations, there is no free parameter and parameters $r^{12}$, $r^{13}$, $r^{14}$ and $r^{23}$ need to be determined. Since $det \mathcal A=a_{11}^3a_{16}^4$, parameters $a_{11},a_{16}$ have to be non-zero. Therefore, from the first and the third equations, $r'^{14}=0$ if and only if $r^{14}=0$, and $r'^{23}=0$ if and only if $r^{23}=0$; in this case $r'^{13}=(a_{11}a_{16})^2r^{13}$ which means $r'^{13}$ has to be the same sign as $r^{13}$, and $r'^{12}=(a_{11}a_{16}r^{12}+a_7r^{13})a_{11}a_{16}^2$, which means $r'^{12}$ can be any arbitrary constant. So, we can consider the following equivalence classes 
\[
\begin{array}{rcl}
&& c^{12}X_1\wedge X_2 +c^{13}X_1\wedge X_3, \quad c^{12}\in \mathbb R, \quad c^{13}\in \mathbb R^{+}-\{0\},\\[2pt]
&& c^{12}X_1\wedge X_2 +c^{13}X_1\wedge X_3, \quad c^{12}\in \mathbb R, \quad c^{13}\in \mathbb R^{-}-\{0\},\\[2pt]
&&  c^{12}X_1\wedge X_2, \quad c^{12}\in \mathbb R.\\
\end{array}
\]
It is obvious that invertible r-matrices can be equivalent to the invertible ones. It is also easy to see from the equations that, if $r^{14}$ and $r^{23}$ are non-zero, then $r'^{14}$ and $r'^{23}$ are non-zero. In this case, $r'^{12}$ and $r'^{13}$ can be arbitrary constants, and equivalence class is
\[
c^{12}X_1\wedge X_2 +c^{13}X_1\wedge X_3+c^{14}X_1\wedge X_4+c^{23}X_2\wedge X_3, \quad c^{12},c^{13}\in \mathbb R, \quad c^{14},c^{23}\in \mathbb R-\{0\}.
\]

\noindent For two other possibilities, $r^{14}=0$ when $r^{23}\neq 0$, and $r^{14}\neq 0$ when $r^{23}= 0$ for which the corresponding $r$-matrices are non-equivalent, we see that $r^{12}$ and $r^{13}$ can be arbitrary constants, so the equivalence classes in this case are
\[
\begin{array}{rcl}
	&& c^{12}X_1\wedge X_2 +c^{13}X_1\wedge X_3+c^{14}X_1\wedge X_4, \quad c^{12},c^{13}\in \mathbb R, \quad c^{14}\in \mathbb R-\{0\},\\[2pt]
&& c^{12}X_1\wedge X_2 +c^{13}X_1\wedge X_3+c^{23}X_2\wedge X_3, \quad c^{12},c^{13}\in \mathbb R, \quad c^{23}\in \mathbb R-\{0\}.
\end{array}
\]

A classification of all $r$-matrices on four-dimensional symplectic real Lie algebras are given in table $1$. The first column gives the names of the Lie algebras according to \cite{Abedi}. The second column gives the non-vanishing structure constants of the Lie algebra $\mathfrak g$. Column three exhibits the equivalence classes of $r$-matrices on  $\mathfrak g$. Column four and five indicate the conditions on some coefficients for each class. The equivalence classes of invertible $r$-matrices are indicated by $\ast$.  

Note that there are two different types of coefficients in this table, free parameters $r^{ij}$ and arbitrary constants $c^{ij}$. We mean by free parameters, the parameters for which different values get non-equivalent $r$-matrices belonging to different equivalence classes; and by arbitrary constants, the parameters such that for every different values of them, the corresponding $r$-matrices are equivalent belonging to the same class.

  {\footnotesize   \bf{Table 1.}} \label{TT}
  {\footnotesize Classification of $r$-matrices on four-dimensional symplectic real Lie algebras. }\\   \begin{tabular}{l l | l l l lp{40mm} }
  	
  	\hline\hline
  	{\scriptsize $\mathfrak g$ }&{\scriptsize  $C_{ij}^k$}   &{\scriptsize  Equivalence classes of $r$-matrices}
  	& {\scriptsize  $$}& {\scriptsize  $$}
  	\smallskip\\
  	\hline
  	\smallskip
  		{\scriptsize $A_{4,1}$}&{\scriptsize  $f_{24}^1=1$}& 
  	{\scriptsize  $\ast \; c^{12}X_1\wedge X_2+c^{13}X_1 \wedge X_3+c^{14}X_1\wedge X_4+c^{23}X_2 \wedge X_3$} &  {\scriptsize$c^{12},c^{13}\in \mathbb R, $} &  {\scriptsize$c^{14},c^{23}\in \mathbb R-\{0\}.$}\\[1pt]
  	
  	{\scriptsize  $$}& {\scriptsize  $f_{34}^2=1$} &{\scriptsize  $c^{12}X_1\wedge X_2+c^{13}X_1 \wedge X_3+c^{23}X_2\wedge X_3$}&   
  	{\scriptsize  $c^{12},c^{13}\in \mathbb R,  $}&
  	{\scriptsize  $c^{23}\in \mathbb R- \{0\}.$}\\[1pt]
  	
  	{\scriptsize  $$}&{\scriptsize  $$}	&
  	{\scriptsize  $c^{12}X_1\wedge X_2+c^{13}X_1 \wedge X_3+c^{14}X_1\wedge X_4 $} &
  	{\scriptsize  $c^{12},c^{13}\in \mathbb R,  $}&
  	{\scriptsize  $c^{14}\in \mathbb R- \{0\}.$}\\[1pt]
  	
  	{\scriptsize  $$}& {\scriptsize  $$} &{\scriptsize  $c^{12}X_1\wedge X_2+c^{13}X_1 \wedge X_3$}&  {\scriptsize  $c^{12}\in \mathbb R,$}  &
  	{\scriptsize  $c^{13}\in  \mathbb R^{+}-\{0\}.$} \\[1pt]
  	
  	{\scriptsize  $$}& {\scriptsize  $$}  &{\scriptsize  $c^{12}X_1\wedge X_2+c^{13}X_1 \wedge X_3$}&  {\scriptsize  $c^{12}\in \mathbb R,$}  &
  	{\scriptsize  $c^{13}\in  \mathbb R^{-}-\{0\}.$} \\[1pt]
  	
  	{\scriptsize  $$}& {\scriptsize  $$} &{\scriptsize  $c^{12}X_1\wedge X_2$} &
  	{\scriptsize  $$}&  {\scriptsize  $c^{12}\in \mathbb R-\{0\}.$}  \\

  	\hline
  	
  	
  	{\scriptsize $A_{4,2}^{-1}$}&{\scriptsize  $f_{14}^1=-1$}& 	 
  	{\scriptsize  $\ast \;c^{12}X_1\wedge X_2+c^{13}X_1 \wedge X_3+c^{23}X_2 \wedge X_3+c^{24}X_2\wedge X_4$} &  {\scriptsize$c^{12}\in \mathbb R, $} &  {\scriptsize$c^{13},c^{24}\in \mathbb R-\{0\},\quad  c^{23}\in \mathbb R^{+}-\{0\}.$}\\[1pt]
  	
  	{\scriptsize $$}&{\scriptsize  $f_{24}^2=1$}& 	 
  	{\scriptsize  $\ast \;c^{12}X_1\wedge X_2+c^{13}X_1 \wedge X_3+c^{23}X_2 \wedge X_3+c^{24}X_2\wedge X_4$} &  {\scriptsize$c^{12}\in \mathbb R, $} &  {\scriptsize$c^{13},c^{24}\in \mathbb R-\{0\},\quad  c^{23}\in \mathbb R^{-}-\{0\}.$}\\[1pt]
  	
  	{\scriptsize $$}&{\scriptsize  $f_{34}^2=1$}& 	
  	{\scriptsize  $c^{12}X_1\wedge X_2+c^{13}X_1 \wedge X_3+c^{23}X_2 \wedge X_3$} &  {\scriptsize$ $} &  {\scriptsize$c^{12},c^{13}\in \mathbb R-\{0\},\quad  c^{23}\in \mathbb R^{+}-\{0\}.$}\\[1pt]
  	
  	{\scriptsize  $$}& {\scriptsize  $f_{34}^3=1$}  &{\scriptsize  $c^{12}X_1\wedge X_2+c^{13}X_1 \wedge X_3+c^{23}X_2\wedge X_3$}&  {\scriptsize$ $} & 
  	{\scriptsize  $c^{12},c^{13}\in \mathbb R-\{0\}, \quad c^{23}\in \mathbb R^{-}- \{0\}.$}\\[1pt]
  	
  	{\scriptsize  $$}& {\scriptsize  $$}  &{\scriptsize  $c^{12}X_1\wedge X_2+c^{23}X_2 \wedge X_3+c^{24}X_2\wedge X_4$}&   
  	{\scriptsize  $c^{12}\in \mathbb R,  $}&
  	{\scriptsize  $c^{24}\in \mathbb R- \{0\},\quad c^{23}\in \mathbb R^{+}- \{0\}.$}\\[1pt]
  	
  	{\scriptsize  $$}& {\scriptsize  $$}  &{\scriptsize  $c^{12}X_1\wedge X_2+c^{23}X_2 \wedge X_3+c^{24}X_2\wedge X_4$}&   
  	{\scriptsize  $c^{12}\in \mathbb R,  $}&
  	{\scriptsize  $c^{24}\in \mathbb R- \{0\},\quad c^{23}\in \mathbb R^{-}- \{0\}.$}\\[1pt]
  	
  	{\scriptsize  $$}&{\scriptsize  $$}	&
  	{\scriptsize  $c^{12}X_1\wedge X_2+c^{13}X_1 \wedge X_3+c^{24}X_2\wedge X_4 $} &
  	{\scriptsize  $c^{12}\in \mathbb R,  $}&
  	{\scriptsize  $c^{13},c^{24}\in \mathbb R- \{0\}.$}\\[1pt]
  	
  	{\scriptsize  $$}&{\scriptsize  $$}	& 
  	{\scriptsize  $c^{12}X_1\wedge X_2+c^{13}X_1\wedge X_3+c^{14}X_1\wedge X_4 $} &
  	{\scriptsize  $c^{12}\in \mathbb R,  $}&
  	{\scriptsize  $c^{13},c^{14}\in \mathbb R- \{0\}.$}\\[1pt]
  	
  	{\scriptsize  $$}& {\scriptsize  $$}  &{\scriptsize  $c^{12}X_1\wedge X_2+c^{13}X_1 \wedge X_3$}&  {\scriptsize  $$}  &
  	{\scriptsize  $c^{12},c^{13}\in  \mathbb R-\{0\}.$} \\[1pt]
  	
  	{\scriptsize  $$}& {\scriptsize  $$}  &{\scriptsize  $c^{13}X_1\wedge X_3+c^{23}X_2 \wedge X_3$}&  {\scriptsize  $$}  &
  	{\scriptsize  $c^{13}\in  \mathbb R-\{0\}, \quad c^{23}\in  \mathbb R^{+}-\{0\}.$} \\[1pt]
  	
  	{\scriptsize  $$}& {\scriptsize  $$}  &{\scriptsize  $c^{13}X_1\wedge X_3+c^{23}X_2 \wedge X_3$}&  {\scriptsize  $$}  &
  	{\scriptsize  $c^{13}\in  \mathbb R-\{0\}, \quad c^{23}\in  \mathbb R^{-}-\{0\}.$} \\[1pt]
  	
  	{\scriptsize  $$}& {\scriptsize  $$}  &{\scriptsize  $c^{12}X_1\wedge X_2+c^{24}X_2 \wedge X_4$}&  {\scriptsize  $c^{12}\in \mathbb R,$}  &
  	{\scriptsize  $c^{24}\in  \mathbb R-\{0\}.$} \\ [1pt]	
  	
  	{\scriptsize  $$}& {\scriptsize  $$}  &{\scriptsize  $c^{12}X_1\wedge X_2+c^{23}X_2 \wedge X_3$}&  {\scriptsize  $$}  &
  	{\scriptsize  $c^{12}\in  \mathbb R-\{0\},\quad c^{23}\in  \mathbb R^{+}-\{0\}.$} \\ 	[1pt]
  	
  	{\scriptsize  $$}& {\scriptsize  $$}  &{\scriptsize  $c^{12}X_1\wedge X_2+c^{23}X_2 \wedge X_3$}&  {\scriptsize  $$}  &
  	{\scriptsize  $c^{12}\in  \mathbb R-\{0\},\quad c^{23}\in  \mathbb R^{-}-\{0\}.$} \\ [1pt]		
  	
  	{\scriptsize  $$}& {\scriptsize  $$}  &{\scriptsize  $c^{12}X_1\wedge X_2+c^{14}X_1 \wedge X_4$}&  {\scriptsize  $c^{12}\in  \mathbb R$}  &
  	{\scriptsize  $c^{14}\in  \mathbb R-\{0\}.$} \\ [1pt]
  	
  	{\scriptsize  $$}& {\scriptsize  $$}  &{\scriptsize  $c^{12}X_1\wedge X_2$}&  {\scriptsize  $$}  &
  	{\scriptsize  $c^{12}\in  \mathbb R-\{0\}.$} \\ 	[1pt]
  	
  	{\scriptsize  $$}& {\scriptsize  $$}  &{\scriptsize  $c^{13}X_1\wedge X_3$}&  {\scriptsize  $$}  &
  	{\scriptsize  $c^{13}\in  \mathbb R-\{0\}.$} \\ [1pt]
  	
  	{\scriptsize  $$}& {\scriptsize  $$}  &{\scriptsize  $c^{23}X_2\wedge X_3$}&  {\scriptsize  $$}  &
  	{\scriptsize  $c^{23}\in  \mathbb R^{+}-\{0\}.$} \\ [1pt]
  	
  	{\scriptsize  $$}& {\scriptsize  $$}  &{\scriptsize  $c^{23}X_2\wedge X_3$}&  {\scriptsize  $$}  &
  	{\scriptsize  $c^{23}\in  \mathbb R^{-}-\{0\}.$} \\ 		
  	\hline
  	
  	{\scriptsize $A_{4,3}$}&{\scriptsize  $f_{14}^1=1$}& 	 
  	{\scriptsize  $\ast \;c^{12}X_1\wedge X_2+c^{13}X_1 \wedge X_3+c^{14}X_1\wedge X_4+c^{23}X_2 \wedge X_3$} &  {\scriptsize$c^{12},c^{13}\in \mathbb R, $} &  {\scriptsize$c^{14},\in \mathbb R-\{0\},\quad  c^{23}\in \mathbb R^{+}-\{0\}.$}\\[1pt]
  	
  	{\scriptsize $$}&{\scriptsize  $f_{34}^2=1$}& 
  	{\scriptsize  $\ast \;c^{12}X_1\wedge X_2+c^{13}X_1 \wedge X_3+c^{14}X_1\wedge X_4+c^{23}X_2 \wedge X_3$} &  {\scriptsize$c^{12},c^{13}\in \mathbb R, $} &  {\scriptsize$c^{14},\in \mathbb R-\{0\},\quad  c^{23}\in \mathbb R^{-}-\{0\}.$}\\[1pt]
  	
  	{\scriptsize  $$}& {\scriptsize  $$}  &{\scriptsize  $c^{12}X_1\wedge X_2+c^{13}X_1 \wedge X_3+c^{23}X_2\wedge X_3$}&   
  	{\scriptsize  $c^{12},c^{13}\in \mathbb R,  $}&
  	{\scriptsize  $c^{23}\in \mathbb R^{+}- \{0\}.$}\\[1pt]
  	
  	{\scriptsize  $$}& {\scriptsize  $$}  &{\scriptsize  $c^{12}X_1\wedge X_2+c^{13}X_1 \wedge X_3+c^{23}X_2\wedge X_3$}&   
  	{\scriptsize  $c^{12},c^{13}\in \mathbb R,  $}&
  	{\scriptsize  $c^{23}\in \mathbb R^{-}- \{0\}.$}\\[1pt]
  	
  	{\scriptsize  $$}&{\scriptsize  $$}	&
  	{\scriptsize  $c^{12}X_1\wedge X_2+c^{13}X_1 \wedge X_3+c^{14}X_1\wedge X_4 $} &
  	{\scriptsize  $c^{12},c^{13}\in \mathbb R,  $}&
  	{\scriptsize  $c^{14}\in \mathbb R- \{0\}.$}\\[1pt]
  	
  	{\scriptsize  $$}&{\scriptsize  $$}	& 
  	{\scriptsize  $c^{12}X_1\wedge X_2+c^{23}X_2 \wedge X_3+c^{24}X_2\wedge X_4 $} &
  	{\scriptsize  $c^{12},c^{23}\in \mathbb R,  $}&
  	{\scriptsize  $c^{24}\in \mathbb R- \{0\}.$}\\[1pt]
  	
  	{\scriptsize  $$}& {\scriptsize  $$}  &{\scriptsize  $c^{12}X_1\wedge X_2+c^{13}X_1 \wedge X_3$}&  {\scriptsize  $c^{12}\in \mathbb R,$}  &
  	{\scriptsize  $c^{13}\in  \mathbb R-\{0\}.$} \\[1pt]
  	
  	{\scriptsize  $$}& {\scriptsize  $$}  &{\scriptsize  $c^{12}X_1\wedge X_2$}&   
  	{\scriptsize  $$}& {\scriptsize  $c^{12}\in \mathbb R-\{0\}.$} \\
  	\hline
  	
  	{\scriptsize $A_{4,6}^{a,0}$}&{\scriptsize  $f_{14}^1=a$}& 	 
  	{\scriptsize  $\ast \;c^{12}X_1\wedge X_2+c^{13}X_1 \wedge X_3+c^{14}X_1\wedge X_4+c^{23}X_2 \wedge X_3$} &  {\scriptsize$c^{12},c^{13}\in \mathbb R, $} &  {\scriptsize$c^{14},\in \mathbb R-\{0\},\quad  c^{23}\in \mathbb R^{+}-\{0\}.$}\\[1pt]
  	
  	{\scriptsize $$}&{\scriptsize  $f_{24}^3=-1$}& 
  	{\scriptsize  $\ast \;c^{12}X_1\wedge X_2+c^{13}X_1 \wedge X_3+c^{14}X_1\wedge X_4+c^{23}X_2 \wedge X_3$} &  {\scriptsize$c^{12},c^{13}\in \mathbb R, $} &  {\scriptsize$c^{14},\in \mathbb R-\{0\},\quad  c^{23}\in \mathbb R^{-}-\{0\}.$}\\[1pt]
  	
  	{\scriptsize  $$}& {\scriptsize  $f_{34}^2=1$}  &{\scriptsize  $c^{12}X_1\wedge X_2+c^{13}X_1 \wedge X_3+c^{23}X_2\wedge X_3$}&   
  	{\scriptsize  $c^{12},c^{13}\in \mathbb R,  $}&
  	{\scriptsize  $c^{23}\in \mathbb R^{+}- \{0\}.$}\\[1pt]
  	
  	{\scriptsize  $$}& {\scriptsize  $a\neq 0$}  &{\scriptsize  $c^{12}X_1\wedge X_2+c^{13}X_1 \wedge X_3+c^{23}X_2\wedge X_3$}&   
  	{\scriptsize  $c^{12},c^{13}\in \mathbb R,  $}&
  	{\scriptsize  $c^{23}\in \mathbb R^{-}- \{0\}.$}\\[1pt]
  	
  	{\scriptsize  $$}&{\scriptsize  $$}	&
  	{\scriptsize  $c^{12}X_1\wedge X_2+c^{13}X_1 \wedge X_3+c^{14}X_1\wedge X_4 $} &
  	{\scriptsize  $c^{12},c^{13}\in \mathbb R,  $}&
  	{\scriptsize  $c^{14}\in \mathbb R- \{0\}.$}\\[1pt]

  	{\scriptsize  $$}& {\scriptsize  $$}  &{\scriptsize  $c^{12}X_1\wedge X_2+c^{13}X_1 \wedge X_3$}&  {\scriptsize  $c^{12},c^{13}\in \mathbb R,$}  &
  	{\scriptsize  $ c^{12} \;\mbox{or} \; c^{13}\neq 0.$} \\[1pt]
 
  	\hline

  	\end{tabular}
  	
  	{\footnotesize   \bf{Table 1.}} \label{TT}
  	{\footnotesize (Continued.) }\\   \begin{tabular}{l l | l l l lp{40mm} }
  	
  	\hline\hline
  	{\scriptsize $\mathfrak g$ }&{\scriptsize  $C_{ij}^k$}   &{\scriptsize  Equivalence classes of $r$-matrices}
  	& {\scriptsize  $$}& {\scriptsize  $$}
  	\smallskip\\
  	\hline
  	\smallskip
  	
  	{\scriptsize $A_{4,7}$}&{\scriptsize  $f_{14}^1=2$}& 	 
  	{\scriptsize  $c^{12}X_1\wedge X_2-c^{23}(\displaystyle \frac{c^{23}}{c^{24}}X_1 \wedge X_3+X_1\wedge X_4-X_2 \wedge X_3)+c^{24}X_2\wedge X_4$} &  {\scriptsize$c^{12},c^{23}\in \mathbb R, $} &  {\scriptsize$c^{24},\in \mathbb R-\{0\}.$}\\[1pt]
  	
  	{\scriptsize $$}&{\scriptsize  $f_{23}^1=1$}& 	 
  	{\scriptsize  $\ast \;c^{12}X_1\wedge X_2+c^{13}X_1 \wedge X_3+c^{23}( \displaystyle\frac{1}{2}X_1\wedge X_4+X_2 \wedge X_3)$} &  {\scriptsize$c^{12},c^{13}\in \mathbb R, $} &  {\scriptsize$c^{23},\in \mathbb R^{+}-\{0\}.$}\\[3pt]
  	
  	{\scriptsize $$}&{\scriptsize  $f_{24}^2=1$}& 
  	{\scriptsize  $c^{12}X_1\wedge X_2+c^{13}X_1 \wedge X_3+c^{23}( \displaystyle\frac{1}{2}X_1\wedge X_4+X_2 \wedge X_3)$} &  {\scriptsize$c^{12},c^{13}\in \mathbb R, $} &  {\scriptsize$c^{23},\in \mathbb R^{-}-\{0\}.$}\\[1pt]
  	
  	{\scriptsize  $$}& {\scriptsize  $f_{34}^2=1$}  &{\scriptsize  $c^{12}X_1\wedge X_2+c^{13}X_1 \wedge X_3+c^{14}X_1\wedge X_4$}&   
  	{\scriptsize  $c^{12},c^{13}\in \mathbb R,  $}&
  	{\scriptsize  $c^{14}\in \mathbb R^{+}- \{0\}.$}\\[1pt]
  	
  	{\scriptsize  $$}& {\scriptsize  $f_{34}^3=1$}  &{\scriptsize  $c^{12}X_1\wedge X_2+c^{13}X_1 \wedge X_3+c^{14}X_1\wedge X_4$}&   
  	{\scriptsize  $c^{12},c^{13}\in \mathbb R,  $}&
  	{\scriptsize  $c^{14}\in \mathbb R^{-}- \{0\}.$}\\[1pt]
  	
  	{\scriptsize  $$}&{\scriptsize  $$}	&
  	{\scriptsize  $c^{12}X_1\wedge X_2+c^{13}X_1 \wedge X_3 $} &
  	{\scriptsize  $c^{12}\in \mathbb R,  $}&
  	{\scriptsize  $c^{13}\in \mathbb R- \{0\}.$}\\[1pt]

  	{\scriptsize  $$}& {\scriptsize  $$}  &{\scriptsize  $c^{12}X_1\wedge X_2$}&  {\scriptsize  $$}  &
  	{\scriptsize  $c^{12}\in \mathbb R- \{0\}.$} \\[1pt]
  	
  	\hline
  	
  	{\scriptsize $A_{4,9}^{-\frac{1}{2}}$}&{\scriptsize  $f_{14}^1=\frac{1}{2}$}& 	 
  	{\scriptsize  $c^{13}X_1\wedge X_3+c^{23}( X_2 \wedge X_3-X_1\wedge X_4+\displaystyle \frac{c^{23}}{c^{34}}X_1 \wedge X_2)+c^{34}X_3\wedge X_4$} &  {\scriptsize$c^{13},c^{23}\in \mathbb R, $} &  {\scriptsize$c^{34},\in \mathbb R-\{0\}.$}\\[1pt]

  	{\scriptsize $$}&{\scriptsize  $f_{23}^1=1$}& 	 
  	{\scriptsize  $\ast \;c^{12}X_1\wedge X_2+c^{13} X_1 \wedge X_3+c^{23}(X_2\wedge X_3+2X_1 \wedge X_4)$} &  {\scriptsize$c^{12},c^{13}\in \mathbb R, $} &  {\scriptsize$c^{23},\in \mathbb R-\{0\}.$}\\[1pt]
  	
  	{\scriptsize $$}&{\scriptsize  $f_{24}^2=1$}& 	 
  	{\scriptsize  $\ast \;c^{12}X_1\wedge X_2+c^{13} X_1 \wedge X_3+c^{24}X_2 \wedge X_4$} &  {\scriptsize$c^{12}\in \mathbb R, $} &  {\scriptsize$c^{13},c^{24} \in \mathbb R-\{0\}.$}\\[1pt]
  	
  	{\scriptsize $$}&{\scriptsize  $f_{34}^3=-\frac{1}{2}$}& 	 
  	{\scriptsize  $c^{12}X_1\wedge X_2+c^{13} X_1 \wedge X_3+c^{14}X_1 \wedge X_4$} &  {\scriptsize$c^{12},c^{13}\in \mathbb R, $} &  {\scriptsize$c^{14} \in \mathbb R-\{0\}.$}\\[1pt]
  	
  	{\scriptsize $$}&{\scriptsize  $$}& 
  	{\scriptsize  $c^{12}X_1\wedge X_2+c^{24}X_2 \wedge X_4$} &  {\scriptsize$c^{12}\in \mathbb R, $} &  {\scriptsize$c^{24} \in \mathbb R-\{0\}.$}\\[1pt]
  	
  	{\scriptsize  $$}& {\scriptsize  $$}  &{\scriptsize  $c^{12}X_1\wedge X_2+c^{13}X_1 \wedge X_3$}&   
  	{\scriptsize  $$}&
  	{\scriptsize  $c^{12},c^{13}\in \mathbb R- \{0\}.$}\\[1pt]
  	
  	{\scriptsize  $$}& {\scriptsize  $$}  &{\scriptsize  $c^{12}X_1\wedge X_2$}&   
  	{\scriptsize  $$}&
  	{\scriptsize  $c^{12}\in \mathbb R- \{0\}.$}\\[1pt]
  	
  	{\scriptsize  $$}& {\scriptsize  $$}  &{\scriptsize  $c^{12}X_1\wedge X_2$}&   
  	{\scriptsize  $$}&
  	{\scriptsize  $c^{12}\in \mathbb R- \{0\}.$}\\[1pt]
  	\hline
  	
  	{\scriptsize $A_{4,9}^{1}$}&{\scriptsize  $f_{14}^1=2$}& 	 
  	{\scriptsize  $\ast \;c^{12}X_1\wedge X_2+c^{13} X_1 \wedge X_3+c^{23}(X_2\wedge X_3+\displaystyle \frac{1}{2}X_1 \wedge X_4)$} &  {\scriptsize$c^{12},c^{13}\in \mathbb R, $} &  {\scriptsize$c^{23} \in \mathbb R-\{0\}.$}\\[2pt]
  	
  	{\scriptsize $$}&{\scriptsize  $f_{23}^1=1$}& 	 
  	{\scriptsize  $c^{13}X_1\wedge X_3+c^{23}(X_2\wedge X_3-X_1\wedge X_4+\displaystyle \frac{c^{23}}{c^{34}}X_1 \wedge X_2)+c^{34}X_3\wedge X_4$} &  {\scriptsize$c^{13},c^{23}\in \mathbb R, $} &  {\scriptsize$c^{34}\in \mathbb R-\{0\}.$}\\[2pt]
  	
  	{\scriptsize $$}&{\scriptsize  $f_{24}^2=1$}& 	 
  	{\scriptsize  $c^{12}X_1\wedge X_2+c^{13}X_1\wedge X_3+c^{14}X_1\wedge X_4$} &  {\scriptsize$c^{12},c^{13}\in \mathbb R, $} &  {\scriptsize$c^{14}\in \mathbb R-\{0\}.$}\\[1pt]
  	
  	{\scriptsize $$}&{\scriptsize  $f_{34}^3=1$}& 	 
  	{\scriptsize  $c^{12}X_1\wedge X_2+c^{13}X_1\wedge X_3$} &  {\scriptsize$c^{12},c^{13} \in \mathbb R$} &  {\scriptsize$c^{12} \;\mbox{or} \; c^{13}\neq 0.$}\\[1pt]	
  	
  	\hline
  	
  	{\scriptsize $A_{4,9}^{0}$}&{\scriptsize  $f_{14}^1=1$}& 	 
  	{\scriptsize  $\ast \; c^{12}X_1\wedge X_2+c^{13}X_1\wedge X_3+c^{23}( X_2 \wedge X_3+X_1\wedge X_4)$} &  {\scriptsize$c^{13}\in \mathbb R, $} &  {\scriptsize$c^{12},c^{23}\in \mathbb R-\{0\}.$}\\[1pt]
  	
  	{\scriptsize $$}&{\scriptsize  $f_{23}^1=1$}& 	 
  	{\scriptsize  $c^{13}X_1\wedge X_3+c^{23}( X_2 \wedge X_3+X_1\wedge X_4)$} &  {\scriptsize$c^{13}\in \mathbb R, $} &  {\scriptsize$c^{23}\in \mathbb R-\{0\}.$}\\[1pt]
  	
  	{\scriptsize $$}&{\scriptsize  $f_{24}^2=1$}& 	 
  	{\scriptsize  $c^{13}X_1\wedge X_3+c^{23}( X_2 \wedge X_3-X_1\wedge X_4+\displaystyle \frac{c^{23}}{c^{34}}X_1\wedge X_2)+c^{34}X_3\wedge X_4$} &  {\scriptsize$c^{13},c^{23}\in \mathbb R, $} &  {\scriptsize$c^{34}\in \mathbb R-\{0\}.$}\\[1pt]
  	
  	{\scriptsize $$}&{\scriptsize  $$}& 
  	{\scriptsize  $c^{12}X_1\wedge X_2+c^{14}X_1\wedge X_4+c^{24}X_2 \wedge X_4$} &  {\scriptsize$c^{12},c^{14}\in \mathbb R, $} &  {\scriptsize$c^{24} \in \mathbb R-\{0\}.$}\\[1pt]
  	
  	{\scriptsize  $$}& {\scriptsize  $$}  &{\scriptsize  $c^{12}X_1\wedge X_2+c^{13}X_1 \wedge X_3+c^{14}X_1\wedge X_4$}&   
  	{\scriptsize  $c^{12}\in \mathbb R,$}&
  	{\scriptsize  $c^{13},c^{14}\in \mathbb R- \{0\}.$}\\[1pt]
  	
  	{\scriptsize  $$}& {\scriptsize  $$}  &{\scriptsize  $c^{12}X_1\wedge X_2+c^{14}X_1\wedge X_4$}&   
  	{\scriptsize  $c^{12}\in \mathbb R,$}&
  	{\scriptsize  $c^{14}\in \mathbb R- \{0\}.$}\\[1pt]
  	
  	{\scriptsize  $$}& {\scriptsize  $$}  &{\scriptsize  $c^{12}X_1\wedge X_2+c^{13}X_1\wedge X_3$}&   
  	{\scriptsize  $$}&
  	{\scriptsize  $c^{12},c^{13}\in \mathbb R- \{0\}.$}\\[1pt]

  	{\scriptsize  $$}& {\scriptsize  $$}  &{\scriptsize  $c^{12}X_1\wedge X_2$}&   
  	{\scriptsize  $$}&
  	{\scriptsize  $c^{12}\in \mathbb R- \{0\}.$}\\[1pt]
  	
  	{\scriptsize  $$}& {\scriptsize  $$}  &{\scriptsize  $c^{13}X_1\wedge X_3$}&   
  	{\scriptsize  $$}&
  	{\scriptsize  $c^{13}\in \mathbb R- \{0\}.$}\\[1pt]
  	\hline
  	
  	
  	{\scriptsize $A_{4,9}^{b}$}&{\scriptsize  $f_{14}^1=1+b$}& 	 
  	{\scriptsize  $\ast \; c^{12}X_1\wedge X_2+c^{13}X_1\wedge X_3+c^{14}( X_1 \wedge X_4+(1+b)X_2\wedge X_3)$} &  {\scriptsize$c^{12},c^{13}\in \mathbb R, $} &  {\scriptsize$c^{14}\in \mathbb R-\{0\}.$}\\[1pt]
  	
  	{\scriptsize $$}&{\scriptsize  $f_{23}^1=1$}& 	 
  	{\scriptsize  $c^{12}X_1\wedge X_2+c^{14}( X_1 \wedge X_4-bX_2\wedge X_3- b\displaystyle\frac{c^{14}}{c^{24}}X_1\wedge X_3)+c^{24}X_2\wedge X_4$} &  {\scriptsize$c^{12},c^{14}\in \mathbb R, $} &  {\scriptsize$c^{24}\in \mathbb R-\{0\}.$}\\[2pt]
  	
  	{\scriptsize $$}&{\scriptsize  $f_{24}^2=1$}& 	 
  	{\scriptsize  $c^{13}X_1\wedge X_3+c^{23}( X_2 \wedge X_3+\displaystyle \frac{c^{23}}{c^{34}}X_1\wedge X_2-X_1\wedge X_4)+c^{34}X_3\wedge X_4$} &  {\scriptsize$c^{13},c^{23}\in \mathbb R, $} &  {\scriptsize$c^{34}\in \mathbb R-\{0\}.$}\\[2pt]
  	
  	{\scriptsize $$}&{\scriptsize  $f_{34}^3=b$}& 	 
  	{\scriptsize  $c^{12}X_1\wedge X_2+c^{13} X_1 \wedge X_3+c^{14} X_1 \wedge X_4$} &  {\scriptsize$c^{12},c^{13}\in \mathbb R $} &  {\scriptsize$c^{14}\in \mathbb R-\{0\}.$}\\[1pt]
  	
  	{\scriptsize $$}&{\scriptsize  $0<\rvert b \lvert <1$}& 	 
  	{\scriptsize  $c^{12}X_1\wedge X_2+c^{13} X_1 \wedge X_3$} &  {\scriptsize$ $} &  {\scriptsize$c^{12},c^{13}\in \mathbb R-\{0\}.$}\\[1pt]
  	
  	{\scriptsize $$}&{\scriptsize  $$}& 	 
  	{\scriptsize  $c^{12}X_1\wedge X_2$} &  {\scriptsize$ $} &  {\scriptsize$c^{12}\in \mathbb R-\{0\}.$}\\[1pt]
  	
  	{\scriptsize $$}&{\scriptsize  $$}& 	 
  	{\scriptsize  $c^{13} X_1 \wedge X_3$} &  {\scriptsize$ $} &  {\scriptsize$c^{13}\in \mathbb R-\{0\}.$}\\[1pt]
  	\hline
  	
  	{\scriptsize $A_{4,11}^{b}$}&{\scriptsize  $f_{14}^1=2b$}& 	 
  	{\scriptsize  $\ast \; c^{12}X_1\wedge X_2+c^{13}X_1\wedge X_3+c^{14}( X_1 \wedge X_4+2bX_2\wedge X_3)$} &  {\scriptsize$c^{12},c^{13}\in \mathbb R, $} &  {\scriptsize$c^{14}\in \mathbb R-\{0\}.$}\\[1pt]
  	
  	{\scriptsize $$}&{\scriptsize  $f_{23}^1=f_{34}^2=1$}& 	 
  	{\scriptsize  $ c^{12}X_1\wedge X_2+c^{13}X_1\wedge X_3+c^{14} X_1 \wedge X_4$} &  {\scriptsize$c^{12},c^{13}\in \mathbb R, $} &  {\scriptsize$c^{14}\in \mathbb R-\{0\}.$}\\[1pt]

  	{\scriptsize $$}&{\scriptsize  $f_{24}^2=f_{34}^3=b$}& 	 
  	{\scriptsize  $ c^{12}X_1\wedge X_2+c^{13}X_1\wedge X_3$} &  {\scriptsize$c^{12},c^{13}\in \mathbb R, $} &  {\scriptsize$c^{12}\;\mbox{or}\; c^{13}\neq 0.$}\\[1pt]
  	
  	{\scriptsize $$}&{\scriptsize  $f_{24}^3=-1$}\\
  	{\scriptsize $$}&{\scriptsize  $b>0$}\\

  	\hline
  	
  	
  \end{tabular}
  
  {\footnotesize   \bf{Table 1.}} \label{TT}
  {\footnotesize (Continued.) }\\   \begin{tabular}{l l | l l l lp{40mm} }
  	
  	\hline\hline
  	{\scriptsize $\mathfrak g$ }&{\scriptsize  $C_{ij}^k$}   &{\scriptsize  Equivalence classes of $r$-matrices}
  	& {\scriptsize  $$}& {\scriptsize  $$}
  	\smallskip\\
  	\hline
  	\smallskip
  	

  	{\scriptsize $A_{4,12}$}&{\scriptsize  $f_{14}^2=-1$}& 	 
  	{\scriptsize  $\ast \; c^{23}( X_2 \wedge X_3-X_1\wedge X_4)+c^{24}( X_2 \wedge X_4+X_1\wedge X_3)+$} &  {\scriptsize$c^{23},c^{24}\in \mathbb R, $} &  {\scriptsize$c^{23}\;\mbox{or}\; c^{24}\neq 0,$}\\[1pt]

  	{\scriptsize $$}&{\scriptsize  $f_{23}^2=1$}& 	 
  	{\scriptsize  $ \quad r^{12}X_1\wedge X_2$} &  {\scriptsize$ $} &  {\scriptsize$r^{12}\neq 0.$}\\[1pt]
  	
  	{\scriptsize $$}&{\scriptsize  $f_{13}^1=1$}& 	 
  	{\scriptsize  $\ast \; c^{23}( X_2 \wedge X_3-X_1\wedge X_4)+c^{24}( X_2 \wedge X_4+X_1\wedge X_3)$} &  {\scriptsize$c^{23},c^{24}\in \mathbb R, $} &  {\scriptsize$c^{23}\;\mbox{or}\; c^{24}\neq 0.$}\\[1pt]

  	{\scriptsize $$}&{\scriptsize  $f_{24}^1=1$}& 	 
  	{\scriptsize  $ c^{13}( X_1 \wedge X_3-X_2\wedge X_4)+c^{23}( X_2 \wedge X_3+X_1\wedge X_4)-$} &  {\scriptsize$c^{13},c^{23}\in \mathbb R, $} &  {\scriptsize$r^{34}\neq 0.$}\\[2pt]

  	{\scriptsize $$}&{\scriptsize  $$}& 	 
  	{\scriptsize  $r^{34}X_3\wedge X_4 -\displaystyle \frac{(c^{13})^2+(c^{23})^2}{r^{34}} X_1\wedge X_2$} &  {\scriptsize$ $} &  {\scriptsize$$}\\[1pt]
  	{\scriptsize $$}&{\scriptsize  $$}& 	 
  	{\scriptsize  $ c^{13} X_1\wedge X_3+c^{23} X_2\wedge X_3$} &  {\scriptsize$ c^{13},c^{23}\in \mathbb R,$} &  {\scriptsize$c^{13}\;\mbox{or}\; c^{23}\neq 0.$}\\[1pt] 
  	
  	{\scriptsize $$}&{\scriptsize  $$}& 	 
  	{\scriptsize  $ c^{12} X_1\wedge X_2$}&  {\scriptsize$$} &  {\scriptsize$ c^{12}\in \mathbb R^{+}-\{0\}.$} \\[1pt]								
  	
  	{\scriptsize $$}&{\scriptsize  $$}& 	 
  	{\scriptsize  $ c^{12} X_1\wedge X_2$} &  {\scriptsize$$}&  {\scriptsize$ c^{12}\in \mathbb R^{-}-\{0\}.$} \\[1pt]						
  	\hline
  	
  	{\scriptsize $A_2\oplus A_2$}&{\scriptsize  $f_{12}^2=1$}& 	 
  	{\scriptsize $ r^{12}(X_1\wedge X_2-X_2\wedge X_3)+c^{14}(X_1 \wedge X_4+X_3 \wedge X_4)+$} & {\scriptsize$ c^{24}\in \mathbb R,$} &  {\scriptsize$c^{14},\in \mathbb R-\{0\}.\quad $}\\
  	
  	{\scriptsize $$}&{\scriptsize  $f_{34}^4=1$}& 	 
  	{\scriptsize $c^{24}X_2 \wedge X_4$} & {\scriptsize$$} &  {\scriptsize$$}\\
  	
  	{\scriptsize  $$}& {\scriptsize  $$}  &{\scriptsize $r^{13}X_1\wedge X_3+c^{23}X_2 \wedge X_3+c^{14}X_1\wedge X_4+\displaystyle \frac {c^{23}c^{14}}{r^{13}}X_2\wedge X_4$}&   
  	{\scriptsize  $ c^{14},c^{23}\in \mathbb R,$}&
  	{\scriptsize  $r^{13} \neq 0.$}\\
  	
  	{\scriptsize $$}&{\scriptsize  $$}& 
  	{\scriptsize  $r^{12}(X_1\wedge X_2-X_2\wedge X_3)+c^{24}X_2\wedge X_4$} &  {\scriptsize$c^{24}\in \mathbb R. $} &  {\scriptsize$$}\\[1pt]

  	{\scriptsize  $$}& {\scriptsize  $$}  &{\scriptsize $\ast \;c^{12}X_1\wedge X_2+c^{24}X_2 \wedge X_4+c^{34}X_3\wedge X_4$}&  {\scriptsize  $ $}&
  	{\scriptsize  $c^{12},c^{24},c^{34}\in \mathbb R- \{0\}.$}\\[1pt]
  	
  	{\scriptsize  $$}& {\scriptsize  $$}  &{\scriptsize $c^{12}X_1\wedge X_2+c^{23}X_2 \wedge X_3+c^{24}X_2\wedge X_4$}&  {\scriptsize  $c^{24}\in \mathbb R, $}& {\scriptsize  $c^{12},c^{23}\in \mathbb R- \{0\}.$}\\[1pt]			
  	
  	{\scriptsize  $$}& {\scriptsize  $$}  &{\scriptsize $c^{14}X_1\wedge X_4+c^{24}X_2 \wedge X_4+c^{34}X_3\wedge X_4$}&  {\scriptsize  $c^{24}\in \mathbb R, $}&
  	{\scriptsize  $c^{14},c^{34}\in \mathbb R- \{0\}.$}\\[1pt] 
  	
  	{\scriptsize  $$}& {\scriptsize  $$}  &{\scriptsize $c^{12}X_1\wedge X_2+c^{24}X_2 \wedge X_4$}&  {\scriptsize  $ $}&
  	{\scriptsize  $c^{12},c^{24}\in \mathbb R- \{0\}.$}\\[1pt]
  	
  	{\scriptsize  $$}& {\scriptsize  $$}  &{\scriptsize $c^{24}X_2\wedge X_4+c^{34}X_3 \wedge X_4$}&  {\scriptsize  $ $}&
  	{\scriptsize  $c^{24},c^{34}\in \mathbb R- \{0\}.$}\\[1pt]	
  	
  	{\scriptsize  $$}& {\scriptsize  $$}  &{\scriptsize $\ast \;c^{12}X_1\wedge X_2+c^{34}X_3 \wedge X_4$}&  {\scriptsize  $ $}&
  	{\scriptsize  $c^{12},c^{34}\in \mathbb R- \{0\}.$}\\[1pt]		
  	
  	{\scriptsize  $$}& {\scriptsize  $$}  &{\scriptsize $c^{23}X_2\wedge X_3+c^{24}X_2 \wedge X_4$}&  {\scriptsize  $c^{24}\in \mathbb R, $}&
  	{\scriptsize  $c^{23}\in \mathbb R- \{0\}.$}\\[1pt]		
  	
  	{\scriptsize  $$}& {\scriptsize  $$}  &{\scriptsize $c^{14}X_1\wedge X_4+c^{24}X_2 \wedge X_4$}&  {\scriptsize  $c^{24}\in \mathbb R, $}&
  	{\scriptsize  $c^{14}\in \mathbb R- \{0\}.$}\\[1pt]	
  	
  	{\scriptsize  $$}& {\scriptsize  $$}  &{\scriptsize $c^{12}X_1\wedge X_2$}&  {\scriptsize  $$}&
  	{\scriptsize  $c^{12}\in \mathbb R- \{0\}.$}\\[1pt]	
  	
  	{\scriptsize  $$}& {\scriptsize  $$}  &{\scriptsize $c^{24}X_2\wedge X_4$}&  {\scriptsize  $$}&
  	{\scriptsize  $c^{24}\in \mathbb R- \{0\}.$}\\[1pt]
  	
  	{\scriptsize  $$}& {\scriptsize  $$}  &{\scriptsize $c^{34}X_3\wedge X_4$}&  {\scriptsize  $$}&
  	{\scriptsize  $c^{34}\in \mathbb R- \{0\}.$}\\[1pt]			
  	\hline
  	
  	{\scriptsize $A_{4,5}^{-1,-1}$}&{\scriptsize  $f_{14}^1=1$}& 	 
  	{\scriptsize  $\ast \;c^{12} X_1 \wedge X_2+ c^{13}X_1\wedge X_3+ c^{23}X_2\wedge X_3+ $} & {\scriptsize$c^{12}\;\mbox{or}\;c^{34}=0,$} &  {\scriptsize$ c^{24},c^{13}\in\mathbb R-\{0\},$}\\[1pt]
  	
  	{\scriptsize $$}&{\scriptsize  $f_{24}^2=-1$}& 	 
  	{\scriptsize  $\quad   c^{24}X_2\wedge X_4+ c^{34}X_3\wedge X_4$} & {\scriptsize$(\mbox{or},\;c^{13}\;\mbox{or}\;c^{24}=0,$} &  {\scriptsize$c^{12},c^{34}\in\mathbb R-\{0\}).$}\\[1pt]

  	{\scriptsize $$}&{\scriptsize  $f_{34}^3=-1$}& 	 
  	{\scriptsize  $c^{12} X_1 \wedge X_2+ c^{13}X_1\wedge X_3+ c^{23}X_2\wedge X_3$} & {\scriptsize$c^{23}\in \mathbb R-\{0\},$} &  {\scriptsize$ c^{12}\; \mbox{or}\;c^{13}\neq 0.$}\\[1pt]
  	
  	{\scriptsize $$}&{\scriptsize  $$}& 	 
  	{\scriptsize  $ c^{23} X_2 \wedge X_3+ c^{24}X_2\wedge X_4+c^{34} X_3\wedge X_4$} & {\scriptsize$c^{23},c^{24},c^{34}\in \mathbb R, $} &  {\scriptsize$c^{24}\; \mbox{or}\;c^{34}\neq 0.$}\\[1pt]
  	
  	{\scriptsize $$}&{\scriptsize  $$}& 	 
  	{\scriptsize  $ c^{12} X_1 \wedge X_2+ c^{13}X_1\wedge X_3+c^{14} X_1\wedge X_4$} & {\scriptsize$c^{12},c^{13}\in \mathbb R, $} &  {\scriptsize$c^{14}\in \mathbb R-\{0\}.$}\\[1pt]
  	
  	{\scriptsize $$}&{\scriptsize  $$}& 	 
  	{\scriptsize  $ c^{12} X_1 \wedge X_2+ c^{13}X_1\wedge X_3$} & {\scriptsize$c^{12},c^{13}\in \mathbb R, $} &  {\scriptsize$c^{12}\; \mbox{or}\;c^{13}\neq 0.$}\\[1pt]

  	{\scriptsize $$}&{\scriptsize  $$}& 	 
  	{\scriptsize  $  c^{23}X_2\wedge X_3$} & {\scriptsize$ $} &  {\scriptsize$c^{23}\in \mathbb R-\{0\}.$}\\[1pt]		
  	\hline
  	
  	{\scriptsize $A_{4,5}^{a,-a}$}&{\scriptsize  $f_{14}^1=1$}& 	 
  	{\scriptsize  $\ast \; c^{12} X_1 \wedge X_2+ c^{13}X_1\wedge X_3+c^{14} X_1\wedge X_4+ c^{23}X_2\wedge X_3$} & {\scriptsize$c^{12},c^{13}\in \mathbb R, $} &  {\scriptsize$c^{14},c^{23}\in \mathbb R-\{0\}.$}\\[1pt]
  	
  	{\scriptsize $$}&{\scriptsize  $f_{24}^2=a$}& 	 
  	{\scriptsize  $ c^{12} X_1 \wedge X_2+ c^{13}X_1\wedge X_3+c^{14} X_1\wedge X_4$} & {\scriptsize$c^{12},c^{13}\in \mathbb R, $} &  {\scriptsize$c^{14}\in \mathbb R-\{0\}.$}\\[1pt]
  	
  	{\scriptsize $$}&{\scriptsize  $f_{34}^3=-a$}& 	 
  	{\scriptsize  $ c^{12} X_1 \wedge X_2+ c^{23}X_2\wedge X_3+c^{24} X_2\wedge X_4$} & {\scriptsize$c^{12},c^{23}\in \mathbb R, $} &  {\scriptsize$c^{24}\in \mathbb R-\{0\}.$}\\[1pt]	
  	
  	{\scriptsize $$}&{\scriptsize  $-1< a<1$}& 	 
  	{\scriptsize  $ c^{13} X_1 \wedge X_3+ c^{23}X_2\wedge X_3+c^{34} X_3\wedge X_4$} & {\scriptsize$c^{13},c^{23}\in \mathbb R, $} &  {\scriptsize$c^{34}\in \mathbb R-\{0\}.$}\\[1pt]
  	
  	{\scriptsize $$}&{\scriptsize  $a\neq 0$}& 	 
  	{\scriptsize  $ c^{12} X_1 \wedge X_2+ c^{13}X_1\wedge X_3+c^{23} X_2\wedge X_3$} & {\scriptsize$ $} &  {\scriptsize$c^{12},c^{13},c^{23}\in \mathbb R-\{0\}.$}\\[1pt]
  	
  	{\scriptsize $$}&{\scriptsize  $$}& 	 
  	{\scriptsize  $ c^{12} X_1 \wedge X_2+ c^{23}X_2\wedge X_3$} & {\scriptsize$ $} &  {\scriptsize$c^{12},c^{23}\in \mathbb R-\{0\}.$}\\[1pt]	
  	
  	{\scriptsize $$}&{\scriptsize  $$}& 	 
  	{\scriptsize  $ c^{13} X_1 \wedge X_3+ c^{23}X_2\wedge X_3$} & {\scriptsize$ $} &  {\scriptsize$c^{13},c^{23}\in \mathbb R-\{0\}.$}\\[1pt]
  	
  	{\scriptsize $$}&{\scriptsize  $$}& 	 
  	{\scriptsize  $ c^{12} X_1 \wedge X_2+ c^{13}X_1\wedge X_3$} & {\scriptsize$ $} &  {\scriptsize$c^{12},c^{13}\in \mathbb R-\{0\}.$}\\[1pt]

  	{\scriptsize $$}&{\scriptsize  $$}& 	 
  	{\scriptsize  $  c^{23}X_2\wedge X_3$} & {\scriptsize$ $} &  {\scriptsize$c^{23}\in \mathbb R-\{0\}.$}\\[1pt]	
  	
  	{\scriptsize $$}&{\scriptsize  $$}& 	 
  	{\scriptsize  $  c^{12}X_1\wedge X_2$} & {\scriptsize$ $} &  {\scriptsize$c^{12}\in \mathbb R-\{0\}.$}\\[1pt]	
  	
  	{\scriptsize $$}&{\scriptsize  $$}& 	 
  	{\scriptsize  $  c^{13}X_1\wedge X_3$} & {\scriptsize$ $} &  {\scriptsize$c^{13}\in \mathbb R-\{0\}.$}\\[1pt]		
  	\hline	
  	
  	\end{tabular}
  	{\footnotesize   \bf{Table 1.}} \label{TT}
  	{\footnotesize (Continued.)}\\   \begin{tabular}{l l | l l l lp{40mm} }
  	
  	\hline\hline
  	{\scriptsize $\mathfrak g$ }&{\scriptsize  $C_{ij}^k$}   &{\scriptsize  Equivalence classes of $r$-matrices}
  	& {\scriptsize  $$}& {\scriptsize  $$}
  	\smallskip\\
  	\hline
  	\smallskip
  	
   	{\scriptsize $A_{4,5}^{-1,a}$}&{\scriptsize  $f_{14}^1=1$}& 	 
   	{\scriptsize  $\ast \; c^{12} X_1 \wedge X_2+ c^{13}X_1\wedge X_3+ c^{23}X_2\wedge X_3+c^{34} X_3\wedge X_4$} & {\scriptsize$c^{12},c^{23}\in \mathbb R, $} &  {\scriptsize$c^{12},c^{34}\in \mathbb R-\{0\}.$}\\[1pt]
   	
   	{\scriptsize $$}&{\scriptsize  $f_{24}^2=-1$}& 	 
   	{\scriptsize  $ c^{12} X_1 \wedge X_2+ c^{13}X_1\wedge X_3+c^{14} X_1\wedge X_4$} & {\scriptsize$c^{12},c^{13}\in \mathbb R, $} &  {\scriptsize$c^{14}\in \mathbb R-\{0\}.$}\\[1pt]
   	
   	{\scriptsize $$}&{\scriptsize  $f_{34}^3=a$}& 	 
   	{\scriptsize  $ c^{12} X_1 \wedge X_2+ c^{23}X_2\wedge X_3+c^{24} X_2\wedge X_4$} & {\scriptsize$c^{12},c^{23}\in \mathbb R, $} &  {\scriptsize$c^{24}\in \mathbb R-\{0\}.$}\\[1pt]	
   	
   	{\scriptsize $$}&{\scriptsize  $-1< a<1$}& 	 
   	{\scriptsize  $ c^{13} X_1 \wedge X_3+ c^{23}X_2\wedge X_3+c^{34} X_3\wedge X_4$} & {\scriptsize$c^{13},c^{23}\in \mathbb R, $} &  {\scriptsize$c^{34}\in \mathbb R-\{0\}.$}\\[1pt]
   	
   	{\scriptsize $$}&{\scriptsize  $a\neq 0$}& 	 
   	{\scriptsize  $ c^{12} X_1 \wedge X_2+ c^{13}X_1\wedge X_3+c^{23} X_2\wedge X_3$} & {\scriptsize$ $} &  {\scriptsize$c^{12},c^{13},c^{23}\in \mathbb R-\{0\}.$}\\[1pt]
   	
   	{\scriptsize $$}&{\scriptsize  $$}& 	 
   	{\scriptsize  $ c^{12} X_1 \wedge X_2+ c^{23}X_2\wedge X_3$} & {\scriptsize$ $} &  {\scriptsize$c^{12},c^{23}\in \mathbb R-\{0\}.$}\\[1pt]	
   	
   	{\scriptsize $$}&{\scriptsize  $$}& 	 
   	{\scriptsize  $ c^{13} X_1 \wedge X_3+ c^{23}X_2\wedge X_3$} & {\scriptsize$ $} &  {\scriptsize$c^{13},c^{23}\in \mathbb R-\{0\}.$}\\[1pt]
   	
   	{\scriptsize $$}&{\scriptsize  $$}& 	 
   	{\scriptsize  $ c^{12} X_1 \wedge X_2+ c^{13}X_1\wedge X_3$} & {\scriptsize$ $} &  {\scriptsize$c^{12},c^{13}\in \mathbb R-\{0\}.$}\\[1pt]

   	{\scriptsize $$}&{\scriptsize  $$}& 	 
   	{\scriptsize  $  c^{23}X_2\wedge X_3$} & {\scriptsize$ $} &  {\scriptsize$c^{23}\in \mathbb R-\{0\}.$}\\[1pt]	
   	
   	{\scriptsize $$}&{\scriptsize  $$}& 	 
   	{\scriptsize  $  c^{12}X_1\wedge X_2$} & {\scriptsize$ $} &  {\scriptsize$c^{12}\in \mathbb R-\{0\}.$}\\[1pt]	
   	
   	{\scriptsize $$}&{\scriptsize  $$}& 	 
   	{\scriptsize  $  c^{13}X_1\wedge X_3$} & {\scriptsize$ $} &  {\scriptsize$c^{13}\in \mathbb R-\{0\}.$}\\[1pt]		
   	\hline	
 {\scriptsize $II\oplus\mathbb R$}&{\scriptsize  $f_{23}^1=1$}& 	 
 {\scriptsize  $\ast \; c^{12} X_1 \wedge X_2+ c^{13}X_1\wedge X_3+ c^{14}X_1\wedge X_4+$} & {\scriptsize$c^{13},c^{14},c^{24}\in \mathbb R, $} &  {\scriptsize$c^{12},c^{34}\in \mathbb R-\{0\},$}\\[1pt]
 
 	{\scriptsize $$}&{\scriptsize  $$}& 	 
 	{\scriptsize  $\quad \quad  c^{24}X_2\wedge X_4+c^{34} X_3\wedge X_4$} & {\scriptsize$(or \; c^{12},c^{14},c^{34}\in \mathbb R, $} &  {\scriptsize$c^{13},c^{24}\in \mathbb R-\{0\}).$}\\[1pt]
 	
 	 {\scriptsize $$}&{\scriptsize  $$}& 	 
 	 {\scriptsize  $ c^{12} X_1 \wedge X_2+ c^{13}X_1\wedge X_3+ c^{14}X_1\wedge X_4$} & {\scriptsize$c^{12},c^{13},c^{14}\in \mathbb R, $} &  {\scriptsize$c^{12} \;\mbox{or}\;c^{13}\neq 0.$}\\[1pt]
 	 
 	  {\scriptsize $$}&{\scriptsize  $$}& 	 
 	  {\scriptsize  $ c^{12} X_1 \wedge X_2+ c^{14}X_1\wedge X_4+ c^{24}X_2\wedge X_4$} & {\scriptsize$c^{12},c^{14} \in \mathbb R, $} &  {\scriptsize$c^{24}\in \mathbb R-\{0\}.$}\\[1pt]
 	  
 	    {\scriptsize $$}&{\scriptsize  $$}& 	 
 	    {\scriptsize  $  c^{14}X_1\wedge X_4$} & {\scriptsize$ $} &  {\scriptsize$c^{14}\in \mathbb R-\{0\}.$}\\[1pt]
 \hline
  {\scriptsize $VI_0\oplus \mathbb R$}&{\scriptsize  $f_{13}^1=1$}& 	 
  {\scriptsize  $\ast \; c^{12} X_1 \wedge X_2+ c^{14}X_1\wedge X_4+c^{24}X_2\wedge X_4+$} & {\scriptsize$c^{14},c^{24}\in \mathbb R, $} &  {\scriptsize$c^{12},c^{34}\in \mathbb R-\{0\}.$}\\[1pt]
  
  {\scriptsize $$}&{\scriptsize  $f_{23}^2=-1$}& 	 
  {\scriptsize  $\quad \quad  c^{34}X_3\wedge X_4$} & {\scriptsize$ $} &  {\scriptsize$$}\\[1pt]
  
  {\scriptsize $$}&{\scriptsize  $$}& 	 
  {\scriptsize  $ c^{12} X_1 \wedge X_2+ c^{13}X_1\wedge X_3+ c^{14}X_1\wedge X_4$} & {\scriptsize$c^{12},c^{14}\in \mathbb R, $} &  {\scriptsize$c^{13}\in \mathbb R-\{0\}.$}\\[1pt]
  
  {\scriptsize $$}&{\scriptsize  $$}& 	 
  {\scriptsize  $ c^{12} X_1 \wedge X_2+ c^{14}X_1\wedge X_4+ c^{24}X_2\wedge X_4$} & {\scriptsize$$} &  {\scriptsize$c^{12},c^{14} ,c^{24}\in \mathbb R-\{0\}.$}\\[1pt]
  
  {\scriptsize $$}&{\scriptsize  $$}& 	 
  {\scriptsize  $ c^{12} X_1 \wedge X_2+ c^{23}X_2\wedge X_3+ c^{24}X_2\wedge X_4$} & {\scriptsize$c^{12},c^{24}\in \mathbb R$} &  {\scriptsize$c^{23}\in \mathbb R-\{0\}.$}\\[1pt]
  
  {\scriptsize $$}&{\scriptsize  $$}& 	 
  {\scriptsize  $ c^{14} X_1 \wedge X_4+ c^{24}X_2\wedge X_4+ c^{34}X_3\wedge X_4$} & {\scriptsize$c^{24},c^{34}\in \mathbb R$} &  {\scriptsize$c^{34}\in \mathbb R-\{0\}.$}\\[1pt]
  
    {\scriptsize $$}&{\scriptsize  $$}& 	 
    {\scriptsize  $ c^{12} X_1 \wedge X_2+ c^{14}X_1\wedge X_4$} & {\scriptsize$$} &  {\scriptsize$c^{12},c^{14}\in \mathbb R-\{0\}.$}\\[1pt]

     {\scriptsize $$}&{\scriptsize  $$}& 	 
     {\scriptsize  $ c^{14} X_1 \wedge X_4+ c^{24}X_2\wedge X_4$} & {\scriptsize$$} &  {\scriptsize$c^{14},c^{24}\in \mathbb R-\{0\}.$}\\[1pt]
     
       {\scriptsize $$}&{\scriptsize  $$}& 	 
       {\scriptsize  $ c^{12} X_1 \wedge X_2+ c^{24}X_2\wedge X_4$} & {\scriptsize$$} &  {\scriptsize$c^{12},c^{24}\in \mathbb R-\{0\}.$}\\[1pt]
       
         {\scriptsize $$}&{\scriptsize  $$}& 	 
         {\scriptsize  $ c^{12} X_1 \wedge X_2$} & {\scriptsize$$} &  {\scriptsize$c^{12}\in \mathbb R-\{0\}.$}\\[1pt]
         
          {\scriptsize $$}&{\scriptsize  $$}& 	 
          {\scriptsize  $ c^{14} X_1 \wedge X_4$} & {\scriptsize$$} &  {\scriptsize$c^{14}\in \mathbb R-\{0\}.$}\\[1pt]
          
           {\scriptsize $$}&{\scriptsize  $$}& 	 
           {\scriptsize  $ c^{24} X_2 \wedge X_4$} & {\scriptsize$$} &  {\scriptsize$c^{24}\in \mathbb R-\{0\}.$}\\[1pt]
            \hline
            {\scriptsize $VII_0\oplus \mathbb R$}&{\scriptsize  $f_{13}^2=-1$}& 	 
            {\scriptsize  $\ast \; c^{12} X_1 \wedge X_2+ c^{14}X_1\wedge X_4+ c^{24}X_2\wedge X_4+$} & {\scriptsize$c^{14},c^{24}\in \mathbb R, $} &  {\scriptsize$c^{12}\in \mathbb R^{+}-\{0\},$}\\[1pt]
            
             {\scriptsize $$}&{\scriptsize  $f_{23}^1=1$}& 	 
             {\scriptsize  $\quad c^{34} X_3 \wedge X_4$} & {\scriptsize$ $} &  {\scriptsize$c^{34}\in \mathbb R-\{0\}.$}\\[1pt]
             
  {\scriptsize $$}&{\scriptsize  $$}& 	 
  {\scriptsize  $\ast \; c^{12} X_1 \wedge X_2+ c^{14}X_1\wedge X_4+ c^{24}X_2\wedge X_4+$} & {\scriptsize$c^{14},c^{24}\in \mathbb R, $} &  {\scriptsize$c^{12}\in \mathbb R^{-}-\{0\},$}\\[1pt]
  
  {\scriptsize $$}&{\scriptsize  $$}& 	 
  {\scriptsize  $\quad c^{34} X_3 \wedge X_4$} & {\scriptsize$ $} &  {\scriptsize$c^{34}\in \mathbb R-\{0\}.$}\\[1pt]

 {\scriptsize $$}&{\scriptsize  $$}& 	 
             {\scriptsize  $ c^{12} X_1 \wedge X_2+ c^{14}X_1\wedge X_4+ c^{24}X_2\wedge X_4$} & {\scriptsize$c^{14}\;\mbox{or}\;c^{24}\neq 0, $} &  {\scriptsize$c^{12}\in \mathbb R^{+}-\{0\},$}\\[1pt]      
             
              {\scriptsize $$}&{\scriptsize  $$}& 	 
              {\scriptsize  $ c^{12} X_1 \wedge X_2+ c^{14}X_1\wedge X_4+ c^{24}X_2\wedge X_4$} & {\scriptsize$c^{14}\;\mbox{or}\;c^{24}\neq 0, $} &  {\scriptsize$c^{12}\in \mathbb R^{-}-\{0\},$}\\[1pt] 
              
               {\scriptsize $$}&{\scriptsize  $$}& 	 
               {\scriptsize  $  c^{14}X_1\wedge X_4+ c^{24}X_2\wedge X_4+c^{34} X_3 \wedge X_4$} & {\scriptsize$c^{14},c^{24}\in \mathbb R, $} &  {\scriptsize$c^{34}\in \mathbb R-\{0\}.$}\\[1pt]
               
   {\scriptsize $$}&{\scriptsize  $$}& 	 
   {\scriptsize  $  c^{14}X_1\wedge X_4+ c^{24}X_2\wedge X_4$} & {\scriptsize$c^{14},c^{24}\in \mathbb R,$} &  {\scriptsize$c^{14}\;\mbox{or}\;c^{24}\neq 0. $}\\[1pt]   
   
   {\scriptsize $$}&{\scriptsize  $$}& 	 
   {\scriptsize  $  c^{12}X_1\wedge X_2$} & {\scriptsize$ $} &  {\scriptsize$c^{12}\in \mathbb R^{+}-\{0\}.$}\\[1pt] 
    {\scriptsize $$}&{\scriptsize  $$}& 	 
    {\scriptsize  $  c^{12}X_1\wedge X_2$} & {\scriptsize$ $} &  {\scriptsize$c^{12}\in \mathbb R^{-}-\{0\}.$}\\[1pt]        
  \hline  
  \hline
   	
   	\end{tabular} 

\subsection {\bf  $r$-$n$ structures with invertible $r$-matrices}\label{all-r-n}
 
 We take the representative  $r=X_1\wedge X_4-X_2\wedge X_3$ of equivalence class of the invertible $r$-matrices on Lie algebra $A_{4,1}$. By inserting $\mathcal{X}_i$ and $\mathcal{Y}^i$, the matrix forms of a generic $(1,1)$-tensor field $n=\sum_{i,j=1}^{4}n^i_{j}X_i\otimes X^j$ and $r$-matrix $r$ in the relations  (\ref{Nijenhuis}), (\ref{Con1}) and (\ref{Con2}), we find all possible Nijenhuis structures which are compatible with $r$.
 
 All $r$-$n$ structures $(r,n)$ on this Lie algebra corresponding to the $r$-matrix $r$ are obtained as the following matrix form
  \begin{equation}\label{Ex}
 \footnotesize{
 	r=\left(\begin{array}{cccc}
 	0& 0& 0 & 1\\ 0& 0& -1 & 0\\ 0& 1& 0 & 0\\ -1& 0& 0 & 0
 	\end{array} \right),\quad
 	n=\left(\begin{array}{cccc}
 	n_1& -n_2& n_4& 0\\ 0& n_3& 0 & n_4\\ 0& 0& n_3 & n_2\\ 0& 0& 0 & n_1
 	\end{array} \right).}
 \end{equation}
 For simplicity, in this case we take $n_1:=n^1_1,n_2:=n^3_4,n_3:=n^2_2,n_4:=n^2_4$.

 A list of $r$-$n$ structures with invertible $r$ on four-dimensional symplectic real Lie algebras are given in table $2$. The first column gives the names of the Lie bialgebras \footnote{We use the notations of the four-dimensional symplectice real Lie bialgebras denoted in \cite{Abedi}.} which are identified by the corresponding $r$-matrices in column three. The second column gives the non-vanishing structure constants of the Sklyanin bracket on the dual Lie algebra $\mathfrak g^*$. Column three exhibits the representatives of the equivalence class of invertible $r$-matrices on  $\mathfrak g$. In column four, the Nijenhuis structures compatible with the corresponding $r$-matrices are presented.
 
 {\footnotesize   \bf{Table 2.}} \label{TT}
 {\footnotesize $r$-$n$ structures with invertible $r$-matrix on four-dimensional symplectic real Lie algebras. }\\   \begin{tabular}{l | l | l l p{40mm} }
 
 \hline\hline
 {\scriptsize $\mathfrak g$ }&{\scriptsize  $ \mbox{Non-zero structure }$}&{\scriptsize  $\mbox{Invertible $r$-matrix $r$}$}
 & {\scriptsize  $\mbox{Nijenhuis structures compatible with $r$}$}\\
 {\scriptsize $\mathfrak g^*$ }&{\scriptsize  $\mbox{constants of $\mathfrak g^*$}$}&{\scriptsize  $$}
 & {\scriptsize  $$}
 \smallskip\\
 \hline
 \smallskip
 
 {\scriptsize $A_{4,1}$}& {\scriptsize ${ \tilde f}^{12}_3={\tilde f}^{13}_4=1$}&
 {\scriptsize  $X_1 \wedge X_4-X_2 \wedge X_3$} &  {\scriptsize$n(X_1)=n_1X_1,\quad n(X_2)=-n_2X_1+n_3X_2,$} \\
 
 {\scriptsize  $A_{4,1}.iii$}  &{\scriptsize  $$}& {\scriptsize  $ $}  &   {\scriptsize  $n(X_3)=n_4X_1+n_3X_3,\quad n(X_4)=n_4X_2+n_2X_3+n_1X_4.$}  \\
 
 \hline
 
 {\scriptsize $A_{4,2}^{-1}$}&{\scriptsize  ${\tilde f}^{12}_1={ \tilde f}^{23}_3=-1 $}&
 {\scriptsize  $X_1\wedge X_3-X_2 \wedge X_4$} & {\scriptsize$n(X_1)=n_1X_1+n_2X_2 ,\quad n(X_2)=n_3X_2,$} \\
 
 {\scriptsize  $A_{4,2}^{-1}.i$}&{\scriptsize  ${ \tilde f}^{24}_4={\tilde  f}^{21}_4=1 $} &  {\scriptsize  $$}&  {\scriptsize  $ n(X_3)=n_4X_2+n_1X_3,\quad n(X_4)=n_4X_1-n_2X_3+n_3X_4.$}\\
 \hline
 {\scriptsize $A_{4,3}$}&{\scriptsize  ${\tilde f}^{12}_3=-1 $}&
 {\scriptsize  $-X_1 \wedge X_4+X_2 \wedge X_3$} & {\scriptsize$n(X_1)=n_1X_1,\quad \quad  n(X_2)=n_2X_1+n_4X_2,$} \\
 
 {\scriptsize  $A_{4,3}.ii$}& {\scriptsize  ${\tilde f}^{14}_4=1 $} & {\scriptsize  $ $} & {\scriptsize  $ n(X_3)=n_3X_1+n_4X_3,\quad n(X_4)=n_3X_2-n_2X_3+n_1X_4.$}  \\
 
 \hline	
 
 {\scriptsize $A_{4,6}^{a,0}$}&{\scriptsize  ${\tilde f}^{21}_3={\tilde f}^{13}_2=-1 $}&
 {\scriptsize  $X_1 \wedge X_4+X_2 \wedge X_3$} & {\scriptsize$n(X_1)=n_1X_1,\quad n(X_2)=n_2X_1+n_4X_2,$} \\
 
 {\scriptsize  $A_{4,6}^{a,0}.i$}& {\scriptsize  ${\tilde f}^{14}_4=-a$} & {\scriptsize  $$} &  {\scriptsize  $ n(X_3)=n_3X_1+n_4X_3 ,\quad n(X_4)=-n_3X_2+n_2X_3+n_1X_4.$}   \\
 
 \hline
 
 {\scriptsize $A_{4,7}$}&{\scriptsize  ${\tilde f}^{12}_2={\tilde f}^{13}_3={\tilde f}^{21}_3=\displaystyle \frac{1}{2} $}&
 {\scriptsize  $\displaystyle -\frac{1}{2}X_1 \wedge X_4-X_2 \wedge X_3$} & {\scriptsize$n(X_1)=n_1X_1,\quad n(X_2)=n_2X_1+n_1X_2,$} \\
 
 {\scriptsize  $A_{4,7}.i$}& {\scriptsize  $ {\tilde f}^{14}_4=1,\;{\tilde f}^{23}_4=2 $} & {\scriptsize  $$}&  {\scriptsize  $ n(X_3)=n_3X_1+n_1X_3,\quad n(X_4)=-2n_3X_2+2n_2X_3+n_1X_4.$}   \\
 
 \hline
 
 {\scriptsize $A_{4,9}^{-\frac{1}{2}}$}&{\scriptsize  ${\tilde f}^{12}_2=-2 ,\; {\tilde f}^{13}_3=4$}&
 {\scriptsize  $-4X_1 \wedge X_4-2X_2\wedge X_3$} & {\scriptsize$\displaystyle n(X_1)=n_2X_1+n_1X_2,\quad n(X_3)=-2n_3X_1+n_2X_3+2n_1X_4,$}  \\
 
 {\scriptsize  $A_{4,9}^{-\frac{1}{2}}.iii$}& {\scriptsize  ${\tilde f}^{14}_4=2,\;{\tilde f}^{23}_4=1$} & {\scriptsize  $ $} &{\scriptsize  $n(X_2)=\displaystyle   2n_4X_1+n_2X_2,\quad n(X_4)=\displaystyle n_3X_2+n_4X_3+n_2X_4.$}   \\
 
 \hline
 {\scriptsize $A_{4,9}^1$}&{\scriptsize  ${\tilde f}^{12}_2={\tilde f}^{13}_3=-\displaystyle \frac{1}{4}$} &
 {\scriptsize  $\displaystyle \frac{1}{4}X_1 \wedge X_4+ \frac{1}{2}X_2 \wedge X_3$} & {\scriptsize$\displaystyle n(X_1)=n_1X_1,\quad n(X_2)=n_2X_1+n_1X_2,$} \\
 
 {\scriptsize  $A_{4,9}^1.i$}& {\scriptsize  ${\tilde f}^{23}_4=-1,\;{\tilde f}^{14}_4=-\displaystyle \frac{1}{2}$} & {\scriptsize  $$} &  {\scriptsize  $\displaystyle  n(X_3)=n_3X_1+n_1X_3,\quad n(X_4)=-2n_3X_2+2n_2X_3+n_1X_4.$}  \\
 \hline
  {\scriptsize $A_{4,9}^0$}&{\scriptsize  ${\tilde f}^{12}_4={\tilde f}^{14}_4=1$} &
  {\scriptsize  $\displaystyle -\frac{1}{2}X_1 \wedge X_2-X_1 \wedge X_4-$} & {\scriptsize$\displaystyle n(X_1)=n_1X_1,\quad  n(X_4)=n_2X_1+n_3X_2-2n_2X_3+n_1X_4,$} \\
  
  {\scriptsize  $A_{4,9}^0.iv$}& {\scriptsize  ${\tilde f}^{23}_4={\tilde f}^{13}_3=1$} & {\scriptsize  $\quad X_2 \wedge X_3$} &  {\scriptsize  $\displaystyle  n(X_2)=-2n_2X_1+n_1X_2,\quad n(X_3)=-n_3X_1+n_1X_3.$}  \\
  \hline
 {\scriptsize $A_{4,9}^b$}&{\scriptsize  ${\tilde f}^{14}_4=1+b$}&
 {\scriptsize  $-X_1 \wedge X_4-(1+b)X_2 \wedge X_3$} & {\scriptsize$\displaystyle n(X_1)=n_1X_1 ,\quad n(X_2)=n_2X_1+n_1X_2,$} \\
 
 {\scriptsize  $A_{4,9}^b.i$}& {\scriptsize  ${\tilde f}^{23}_4=(1+b)^2$} & {\scriptsize  $$} &  {\scriptsize  $\displaystyle  n(X_3)=n_3X_1+n_1X_3,$}    \\
 
 {\scriptsize  $$}& {\scriptsize  ${\tilde f}^{12}_2=b,\;{\tilde f}^{13}_3=1$} & {\scriptsize  $$} &  {\scriptsize  $  n(X_4)=(1+b)(n_2X_3-n_3X_2)+n_1X_4.$}    \\
 
 \hline
 {\scriptsize $A_{4,11}^b$}&{\scriptsize  ${\tilde f}^{12}_2={\tilde f}^{13}_3=\displaystyle \frac {1}{2}$}&
 {\scriptsize  $\displaystyle -\frac {1}{2b}X_1 \wedge X_4-X_2 \wedge X_3$} & {\scriptsize$\displaystyle n(X_1)=n_1X_1 ,\quad n(X_2)=n_2X_1+n_1X_2,$} \\

 {\scriptsize  $A_{4,11}^b.i$}& {\scriptsize  ${\tilde f}^{23}_4=2b,\;{\tilde f}^{14}_4=1 $} & {\scriptsize  $$} &  {\scriptsize  $\displaystyle  n(X_3)=n_3X_1+n_1X_3, $}   \\

 {\scriptsize  $$}& {\scriptsize  ${\tilde f}^{21}_3={\tilde f}^{13}_2=\displaystyle \frac {1}{2b} $} & {\scriptsize  $$} & {\scriptsize  $ n(X_4)=-2bn_3X_2+2bn_2X_3+n_1X_4.$}\\
 
 \hline
 {\scriptsize $A_{4,12}$}&{\scriptsize  ${\tilde f}^{23}_3={\tilde f}^{24}_4=-1$}&
 {\scriptsize  $\displaystyle -\frac {b}{2}X_1 \wedge X_2-X_1 \wedge X_4$} & {\scriptsize$\displaystyle n(X_1)=n_1X_1-n_2X_2,\quad n(X_2)=n_2X_1+n_1X_2,$} \\
 
 {\scriptsize  $A_{4,12}.ii$}& {\scriptsize  ${\tilde f}^{14}_3={\tilde f}^{31}_1=1 $} & {\scriptsize  $+X_2 \wedge X_3$} &    {\scriptsize  $\displaystyle  n(X_3)=\displaystyle n_3X_1+\frac{1}{2}bn_2X_2+ n_1X_3+n_2X_4, $}\\
 
 {\scriptsize  $$}&{\scriptsize  ${\tilde f}^{12}_3=b$}& {\scriptsize  $(b\in \mathbb R-\{0\})$}& {\scriptsize  $n(X_4)=\displaystyle -\frac{1}{2}bn_2X_1+n_3X_2-n_2X_3+n_1X_4.$}\\
 \hline
 {\scriptsize $A \oplus A $}&{\scriptsize  ${\tilde f}^{34}_3=1$}&
 {\scriptsize  $X_1\wedge X_2 +X_3 \wedge X_4$} & {\scriptsize$\displaystyle n(X_1)=n_1X_1-n_2X_4,\quad n(X_2)=n_1X_2,$} \\
 
 {\scriptsize  $(III\oplus \mathbb R).iii$}&{\scriptsize  $$} & {\scriptsize  $$}&  {\scriptsize  $\displaystyle  n(X_3)=n_2X_2+n_3X_3,\quad n(X_4)=n_3X_4.$}   \\
 
 \hline

\end{tabular}

 {\footnotesize   \bf{Table 2.}} \label{TT}
 {\footnotesize (Continued.)}\\   \begin{tabular}{l | l | l l p{40mm} }
 	
 	\hline\hline
 	{\scriptsize $\mathfrak g$ }&{\scriptsize  $ \mbox{Non-zero structure }$}&{\scriptsize  $\mbox{Invertible $r$-matrix $r$r}$}
 	& {\scriptsize  $\mbox{Nijenhuis structures compatible with $r$}$}\\
 	{\scriptsize $\mathfrak g^*$ }&{\scriptsize  $\mbox{constants of $\mathfrak g^*$}$}&{\scriptsize  $$}
 	& {\scriptsize  $$}
 	\smallskip\\
 	\hline
 	\smallskip
 {\scriptsize $A_{4,5}^{-1,-1} $}&{\scriptsize  ${\tilde f}^{12}_1={\tilde f}^{23}_3=-1$}&
 {\scriptsize  $X_1\wedge X_2 +X_1 \wedge X_3+$} & {\scriptsize$\displaystyle n(X_1)=(n_2-n_1)X_1+n_5X_2+n_1X_4,$} \\
  
 {\scriptsize  $A_{4,5}^{-1,-1}.i$}&{\scriptsize  ${\tilde f}^{24}_4=1$} & {\scriptsize  $X_2 \wedge X_4$}&  {\scriptsize  $\displaystyle  n(x_2)=(n_3-n_1)X_2-n_1X_3,\quad n(X_3)=n_4X_2+n_2X_3,$}   \\
 
 {\scriptsize  $$}&{\scriptsize  $$} & {\scriptsize  $$}&  {\scriptsize  $ n(X_4)=(n_2-n_3-n_4)X_1+n_5X_2+n_5X_3+n_3X_4.$}   \\
 
 \hline
 {\scriptsize $A_{4,5}^{a,-a} $}&{\scriptsize  ${\tilde f}^{21}_2={\tilde f}^{13}_3=a$}&
 {\scriptsize  $X_2\wedge X_3 -X_1 \wedge X_4$} & {\scriptsize$\displaystyle n(X_1)=n_1X_1,\quad n(X_2)=-n_2X_1+n_4X_2,$} \\
 
 {\scriptsize  $A_{4,5}^{a,-a}.i$}&{\scriptsize  ${\tilde f}^{14}_4=1$} & {\scriptsize  $$}&  {\scriptsize  $\displaystyle  n(X_3)=n_3X_1+n_4X_3,\quad n(X_4)=n_3X_2+n_2X_3+n_1X_4.$}   \\
 
 \hline
 {\scriptsize $A_{4,5}^{-1,a} $}&{\scriptsize  ${\tilde f}^{13}_1={\tilde f}^{32}_2=1$}&
 {\scriptsize  $X_1\wedge X_2 -X_3 \wedge X_4$} & {\scriptsize$\displaystyle n(X_1)=n_1X_1+n_3X_3,\quad n(X_2)=n_1X_2+n_2X_3,$} \\
 
 {\scriptsize  $A_{4,5}^{-1,a}.i$}&{\scriptsize  ${\tilde f}^{34}_4=a$} & {\scriptsize  $$}&  {\scriptsize  $\displaystyle  n(X_3)=n_4X_3,\quad n(X_4)=n_2X_1-n_3X_2+n_4X_4.$}   \\
 
 \hline
 
 {\scriptsize $II \oplus \mathbb R$}&{\scriptsize  ${\tilde f}^{14}_3=1$}&
 {\scriptsize  $X_1 \wedge X_3-X_2 \wedge X_4$} & {\scriptsize$\displaystyle n(X_1)=n_1X_1+n_5X_4,\quad n(X_2)=-n_2X_1+n_4X_2+n_5X_3, $} \\
 
 {\scriptsize  $(II \oplus \mathbb R).xiv$}& {\scriptsize  $$} & {\scriptsize  $$} &  {\scriptsize  $n(X_3)=n_3X_2+n_1X_3+n_2X_4,\quad n(X_4)=n_3X_1+n_4X_4.$}   \\
 \hline
 
 {\scriptsize $VI_0\oplus \mathbb R$}&{\scriptsize  ${\tilde f}^{14}_1=1$}&
 {\scriptsize  $X_1 \wedge X_2+X_3 \wedge X_4$} & {\scriptsize$\displaystyle n(X_1)=n_1X_1+n_3X_4,\quad n(X_2)=n_1X_2+n_2X_4, $} \\
 
 {\scriptsize  $(VI_0\oplus \mathbb R).ix$}& {\scriptsize  ${\tilde f}^{24}_2=-1$} & {\scriptsize  $$} &  {\scriptsize  $n(X_3)=n_2X_1-n_3X_2+n_4X_3,\quad n(X_4)=n_4X_4.$}   \\
 
 \hline  
 {\scriptsize $VII_0 \oplus \mathbb R$}&{\scriptsize  ${\tilde f}^{14}_2=1$}&
 {\scriptsize  $X_1 \wedge X_2+X_3 \wedge X_4$}& {\scriptsize$\displaystyle n(X_1)=n_1X_1+n_3X_4,\quad n(X_2)=n_1X_2+n_2X_4 ,$} \\
 
 {\scriptsize  $(VII_0 \oplus \mathbb R).iv$}& {\scriptsize  ${\tilde f}^{24}_1=-1$} & {\scriptsize  $$}&  {\scriptsize  $n(X_3)=n_2X_1-n_3X_2+n_4X_3,\quad n(X_4)=n_4X_4.$}  \\
 
 \hline  	
\hline 	
 	\end{tabular}

\subsection {\bf  Classification of $r$-$n$ structures with invertible $r$-matrices}\label{class-r-n}
 
 According to Proposition \ref{Pro}, if we find all equivalence classes of Nijenhuis structures such that $(r,n')\sim_0 (r,n) $, then we have all equivalence classes of $r$-$n$ structures corresponding to the $r$-matrix $r$. So, we consider automorphism group element of the Lie algebra $A_{4,1}$ such that $\mathcal A r \mathcal A^t=r$, and we get the solution
 \begin{equation}\label{auto}
   a_{11} = 1,\; a_{12 }= a_7,\; a_{16} = 1, \;a_3 = -a_8+a_7^2.
   \end{equation}
The automorphism $\mathcal A$ in the new expression, given by (\ref{auto}) is
  \begin{equation}\label{auto2}
  \footnotesize{
  	\mathcal A=\left(\begin{array}{cccc}
  	1 & a_7  & a_7^2 -a_8& a_4\\ 0 & 1& a_7 & a_8\\0 & 0 & 1 & a_7\\ 0 & 0 & 0 & 1
  	\end{array} \right)}.
  \end{equation}
  Since $det(\mathcal A)=1\neq 0$ and does not depend on parameters $a_4,a_7,a_8$, these parameters can take any value.
  \noindent Now consider two structures $n$ and $n'$
  \[
  \footnotesize{
  	n=\left(\begin{array}{cccc}
  	n_1& -n_2& n_4& 0\\ 0& n_3& 0 & n_4\\ 0& 0& n_3 & n_2\\ 0& 0& 0 & n_1
  	\end{array} \right),\quad
  	n'=\left(\begin{array}{cccc}
  	n'_1& -n'_2& n'_4& 0\\ 0& n'_3& 0 & n'_4\\ 0& 0& n'_3 & n'_2\\ 0& 0& 0 & n'_1
  	\end{array} \right).}
  \]
  Applying the automorphism group element (\ref{auto2}) in the relation $n\circ \mathcal A-\mathcal A\circ n'=0$, we obtain
  \[
  n_1-n'_1= 0, \quad n_3-n'_3= 0, \quad n'_2-n_2+a_7n_3-a_7n'_1= 0,\quad n_2-n'_2+a_7n_1-a_7n'_3= 0,
  \]
  \[
  n_4-n'_4+(a_7^2-a_8)n_3-(a_7^2-a_8)n'_1+n'_2a_7= 0, \quad n_4-n'_4+a_7n_2+a_8n_1-n'_3a_8= 0,
  \]
  \[
  a_7n_4+(a_7^2-a_8)n_2+a_4n_1-a_4n'_1+a_8n'_2-a_7n'_4= 0,\quad a_7 n_3-a_7n'_3 = 0 ,
  \]

  \noindent  Since $n_1=n'_1$ and $n_3=n'_3$, it means two elements $n_1$ and $n_3$ are free parameters. Parameters $n_2$ and $n_4$ need to be determined. Inserting $n_1=n'_1$ and $n_3=n'_3$ in the above equations we get
  \[
  \begin{array}{rclcrclcrclcrcl}
  && (n_2-n'_2)+a_7(n_1-n_3)= 0, \\[2pt]
  &&(n_4-n'_4)+a_7n_2+a_8(n_1-n_3)=0, \\ [2pt]
  && a_7(n_4-n'_4)+a_7^2n_2-a_8(n_2-n'_2)=0,\\[2pt]
  &&(n_4-n'_4)+(a_7^2-a_8)(n_3-n_1)+n'_2a_7 = 0
  .\\[2pt]
  \end{array}
  \]
  The above equations, can be reduced to the two equations
  \[
  (n_2-n'_2)-a_7(n_1-n_3)= 0,\quad
  (n_4-n'_4)+(a_7^2-a_8)(n_3-n_1)+n'_2a_7= 0.
  \]
  First, if $n_1=n_3$, then  $n_2=n'_2$ and $(n_4-n'_4)+n_2a_7=0$. Since the elements of automorphism group are arbitrary, there are two possibilities, $n_2=0$ which implies $n_4=n'_4$; or $n_2\neq 0$ which implies $n_4=n'_4-n_2a_7$. Therefore, $n_4$ can be any arbitrary constant in the both cases.
  
  \noindent Second, if $n_1\neq n_3$, it is easy to see that $n_2$ and $n_4$ can be any arbitrary constant. Eventually, equivalence classes of Nijenhuis structures $n$ are classified as follows
  \[
  {\footnotesize n^{(1)}=	\left(\begin{array}{cccc}
  	n_1& 0& c_1 & 0\\ 0& n_1& 0 & c_1\\ 0& 0& n_1 & 0\\ 0& 0& 0 & n_1
  	\end{array} \right),\quad
  	n^{(2)}=	\left(\begin{array}{cccc}
  	n_1& -c_2& c_3 & 0\\ 0& n_1& 0 & c_3\\ 0& 0& n_1 & c_2\\ 0& 0& 0 & n_1
  	\end{array} \right),\quad 	
  	n^{(3)}=	\left(\begin{array}{cccc}
  	n_1& -c_4& c_5 & 0\\ 0& n_3& 0 & c_5\\ 0& 0& n_3 & c_4\\ 0& 0& 0 & n_1
  	\end{array} \right),}
  \]
  where $c_1,c_3,c_4,c_5 \in \mathbb R$ and $c_2\in {\mathbb R}-\{0\}$. 
 
  A list of equivalence classes of the given Nijenhuis structures in table 2, for other Lie algebras are given in table 3.
   
Note that, similar to the table 1, there are two different types of coefficients in this table, free parameters $n_i$ and arbitrary constants $c_i$. We mean by free parameters $n_i$ , the parameters for which different values get non-equivalent Nijenhuis structures belonging to the different equivalence classes; and by $c_i$, arbitrary constants such that for every different values of them, the corresponding Nijenhuis structures are equivalent belonging to the same class.

{\footnotesize   \bf{Table 3.}} \label{TT}
{\footnotesize Classification of Nijenhuis structures. }\\   \begin{tabular}{l | l l l l lp{40mm} }
	
	\hline\hline
	{\scriptsize $\mathfrak g$ }&{\scriptsize  $\mbox{Equivalence classes}$}&{\scriptsize $$}
	& {\scriptsize  $$}
	\smallskip\\
	\hline
	\smallskip

	
	{\scriptsize $A_{4,1}$}&	 {\scriptsize  $n(X_1)=n_1X_1, $}&	{\scriptsize  $n(X_2)=n_1X_2,$} & {\scriptsize  $c_1\in\mathbb R.$}\\
	
	{\scriptsize$$}&	 {\scriptsize  $ n(X_3)=c_1X_1+n_1X_3,$}&	{\scriptsize  $n(X_4)=c_1X_2+n_1X_4,$} &  {\scriptsize$$}\\
	\hline
	{\scriptsize $$}& {\scriptsize  $n(X_1)=n_1X_1, $}&	{\scriptsize  $n(X_2)=-c_2X_1+n_1X_2,$} &  {\scriptsize  $c_3\in\mathbb R,$}\\
	
	{\scriptsize $$}& {\scriptsize  $n(X_3)=c_3X_1+n_1X_3,$}&	{\scriptsize  $n(X_4)=c_3X_2+c_2X_3+n_1X_4,$} &	{\scriptsize $ c_2\in\mathbb R-\{0\}.$}\\
	\hline
	{\scriptsize $$}& {\scriptsize  $n(X_1)=n_1X_1, $}&	{\scriptsize  $n(X_2)=-c_4X_1+n_3X_2,$} &{\scriptsize  $c_4,c_5\in\mathbb R,$}\\
	
	{\scriptsize $$}& {\scriptsize  $n(X_3)=c_5X_1+n_3X_3,  $}&	{\scriptsize  $n(X_4)=c_5X_2+c_4X_3+n_1X_4,$} &  {\scriptsize$n_1\neq n_3.$} \\
	\hline
	
	{\scriptsize $A_{4,2}^{-1}$}& {\scriptsize  $n(X_1)=n_1X_1+n_2X_2 , $}&	{\scriptsize  $n(X_2)=n_3X_2 ,$} &{\scriptsize  $$}\\
	
	{\scriptsize $$}& {\scriptsize  $n(X_3)=n_4X_2+n_1X_3 ,  $}&	{\scriptsize  $n(X_4)=n_4X_1-n_2X_3+n_3X_4.$} &  {\scriptsize$$}\\
	\hline

	{\scriptsize $A_{4,3}$}& {\scriptsize  $n(X_1)=n_1X_1, $}&	{\scriptsize  $n(X_2)=n_2X_1+n_4X_2,$} &{\scriptsize  $c_3\in\mathbb R,$}\\
	
	{\scriptsize $$}& {\scriptsize  $n(X_3)=c_3X_1+n_4X_3,  $}&	{\scriptsize  $n(X_4)=c_3X_2-n_2X_3+n_1X_4,$} &  {\scriptsize$n_2\in\mathbb R^{+} (or, n_2\in\mathbb R^{-}).$} \\
	
	\hline
	
	{\scriptsize $A_{4,6}^{a,0}$}& {\scriptsize  $n(X_1)=n_1X_1, $}&	{\scriptsize  $n(X_2)=n_2X_1+n_4X_2,$} &{\scriptsize  $$}\\
	
	{\scriptsize $$}& {\scriptsize  $n(X_3)=n_3X_1+n_4X_3,  $}&	{\scriptsize  $n(X_4)=-bn_3X_2+bn_2X_3+n_1X_4,$} &  {\scriptsize $$} \\
	
	\hline	

	{\scriptsize $A_{4,7}$}& {\scriptsize  $n(X_1)=n_1X_1, $}&	{\scriptsize  $n(X_2)=n_2X_1+n_1X_2,$} &{\scriptsize  $c_3\in\mathbb R,$}\\
	
	{\scriptsize $$}& {\scriptsize  $n(X_3)=c_3X_1+n_1X_3,  $}&	{\scriptsize  $n(X_4)=-2c_3X_2+2n_2X_3+n_1X_4,$} &  {\scriptsize$n_2\in\mathbb R^{+} (or, n_2\in\mathbb R^{-}).$} \\
	
	\hline
	
	{\scriptsize $A_{4,9}^{-\frac{1}{2}}$}& {\scriptsize  $n(X_1)=n_2X_1, $}&	{\scriptsize  $n(X_2)=2c_4X_1+n_2X_2,$} &{\scriptsize  $c_4\in\mathbb R-\{0\}.$}\\
	
	{\scriptsize $$}& {\scriptsize  $n(X_3)=n_2X_3,  $}&	{\scriptsize  $n(X_4)=c_4X_3+n_2X_4,$} &  {\scriptsize$$} \\
	\hline
	
	{\scriptsize $$}& {\scriptsize  $n(X_1)=n_2X_1, $}&	{\scriptsize  $n(X_2)=n_2X_2,$} &{\scriptsize  $c_3\in\mathbb R-\{0\}.$}\\
	
	{\scriptsize $$}& {\scriptsize  $n(X_3)=-2c_3X_1+n_2X_3,  $}&	{\scriptsize  $n(X_4)=c_3X_2+n_2X_4,$} &  {\scriptsize$$} \\
	
	\hline
	
	{\scriptsize $$}& {\scriptsize  $n(X_1)=n_2X_1+c_1X_2, $}&	{\scriptsize  $n(X_2)=2c_4X_1+n_2X_2,$} &{\scriptsize  $c_3\in\mathbb R,$}\\
	
	{\scriptsize $$}& {\scriptsize  $n(X_3)=-2c_3X_1+n_2X_3+2c_1X_4,  $}&	{\scriptsize  $n(X_4)=c_3X_2+c_4X_3+n_2X_4,$} &  {\scriptsize$c_1,c_4\in \mathbb R-\{0\}.$} \\
	
	\hline
	
	{\scriptsize $A_{4,9}^1$}& {\scriptsize  $n(X_1)=n_1X_1, $}&	{\scriptsize  $n(X_2)=c_2X_1+n_1X_2,$} &{\scriptsize  $c_2,c_3\in\mathbb R-\{0\}.$}\\
	
	{\scriptsize $$}& {\scriptsize  $n(X_3)=c_3X_1+n_1X_3,  $}&	{\scriptsize  $n(X_4)=-2c_3X_2+2c_2X_3+n_1X_4,$} &  {\scriptsize$$} \\
	
	\hline
	
	{\scriptsize $$}& {\scriptsize  $n(X_1)=n_1X_1, $}&	{\scriptsize  $n(X_2)=n_2X_1+n_1X_2,$} &{\scriptsize  $n_2\;\mbox{or}\;n_3=0.$}\\
	
	{\scriptsize $$}& {\scriptsize  $n(X_3)=n_3X_1+n_1X_3,  $}&	{\scriptsize  $n(X_4)=-2n_3X_2+2n_2X_3+n_1X_4,$} &  {\scriptsize$$} \\
	
	\hline
		{\scriptsize $A_{4,9}^0$}& {\scriptsize  $n(X_1)=n_1X_1, $}&	{\scriptsize  $n(X_2)=-2n_2X_1+n_1X_2,$} &{\scriptsize  $$}\\
		
		{\scriptsize $$}& {\scriptsize  $n(X_3)=-n_3X_1+n_1X_3,  $}&	{\scriptsize  $n(X_4)=n_2X_1+n_3X_2-2n_2X_3+n_1X_4.$} &  {\scriptsize$$} \\
		
		\hline
		
\end{tabular}

{\footnotesize   \bf{Table 3.}} \label{TT}
{\footnotesize (Continued.) }\\   \begin{tabular}{l | l l l l lp{40mm} }
	
	\hline\hline
	{\scriptsize $\mathfrak g$ }&{\scriptsize  $\mbox{Equivalence classes}$}&{\scriptsize $$}
	& {\scriptsize  $$}
	\smallskip\\
	\hline
	\smallskip

		{\scriptsize $A_{4,9}^b$}& {\scriptsize  $n(X_1)=n_1X_1, $}&	{\scriptsize  $n(X_2)=c_2X_1+n_1X_2,$} &{\scriptsize  $c_2,c_3\in \mathbb R-\{0\}.$}\\
		
		{\scriptsize $$}& {\scriptsize  $n(X_3)=c_3X_1+n_1X_3,  $}&	{\scriptsize  $n(X_4)=-c_3(1+b)X_2+c_2(1+b)X_3+n_1X_4,$} &  {\scriptsize$$} \\
		
		\hline
		
		{\scriptsize $$}& {\scriptsize  $n(X_1)=n_1X_1, $}&	{\scriptsize  $n(X_2)=c_2X_1+n_1X_2,$} &{\scriptsize  $c_2\in \mathbb R-\{0\}.$}\\
		
		{\scriptsize $$}& {\scriptsize  $n(X_3)=n_1X_3,  $}&	{\scriptsize  $n(X_4)=c_2(1+b)X_3+n_1X_4,$} &  {\scriptsize$$} \\
		
		\hline
		
		{\scriptsize $$}& {\scriptsize  $n(X_1)=n_1X_1, $}&	{\scriptsize  $n(X_2)=n_1X_2,$} &{\scriptsize  $c_3\in \mathbb R-\{0\}.$}\\
		
		{\scriptsize $$}& {\scriptsize  $n(X_3)=c_3X_1+n_1X_3,  $}&	{\scriptsize  $n(X_4)=-c_3(1+b)X_2+n_1X_4,$} &  {\scriptsize$$} \\
		\hline
		
		{\scriptsize $A_{4,11}^b$}& {\scriptsize  $n(X_1)=n_1X_1, $}&	{\scriptsize  $n(X_2)=c_2X_1+n_1X_2,$} &{\scriptsize  $c_2,c_3\in \mathbb R,$}\\
		
		{\scriptsize $$}& {\scriptsize  $n(X_3)=c_3X_1+n_1X_3,  $}&	{\scriptsize  $n(X_4)=-2bc_3X_2+2bc_2X_3+n_1X_4,$} &  {\scriptsize$c_2\;\mbox{or} \; c_3 \neq 0.$} \\
		\hline
		
		{\scriptsize $A_{4,12}$}& {\scriptsize  $n(X_1)=n_1X_1-n_2X_2, $}&	{\scriptsize  $n(X_2)=n_2X_1+n_1X_2,$} &{\scriptsize  $$}\\
		
		{\scriptsize $$}& {\scriptsize  $n(X_3)=\displaystyle n_3X_1+\frac{1}{2}bn_2X_2+n_1X_3+n_2X_4,  $}&	{\scriptsize  $n(X_4)=\displaystyle -\frac{1}{2}bn_2X_1+n_3X_2-n_2X_3+n_1X_4.$} &  {\scriptsize$$} \\
		\hline
		{\scriptsize $A_2\oplus A_2$}& {\scriptsize  $n(X_1)=n_1X_1-n_2X_4, $}&	{\scriptsize  $n(X_2)=n_1X_2,$} &{\scriptsize  $$}\\
		
		{\scriptsize $$}& {\scriptsize  $n(X_3)=\displaystyle n_2X_2+n_3X_3,  $}&	{\scriptsize  $n(X_4)=n_3X_4.$} &  {\scriptsize$$} \\
		\hline
	{\scriptsize $A_{4,5}^{-1,-1}$}& {\scriptsize  $n(X_1)=(n_2-n_1)X_1+c_5X_2+n_1X_4, $}&	{\scriptsize  $n(X_2)=(n_3-n_1)X_2-n_1X_3,$} &{\scriptsize  $c_5 \in \mathbb R,$}\\
	
	{\scriptsize $$}& {\scriptsize  $n(X_3)=n_4X_2+n_2X_3,  $}&	{\scriptsize  $n(X_4)=(n_2-n_3-n_4)X_1+c_5X_2+c_5X_3+n_3X_4,$} &  {\scriptsize$n_1\neq 0.$} \\
	
	\hline
	{\scriptsize $$}& {\scriptsize  $n(X_1)=n_2X_1+c_5X_2, $}&	{\scriptsize  $n(X_2)=n_3X_2,$} &{\scriptsize  $c_4 \in \mathbb R,$}\\
	
	{\scriptsize $$}& {\scriptsize  $n(X_3)=c_4X_2+n_2X_3,  $}&	{\scriptsize  $n(X_4)=(n_2-n_3-c_4)X_1+c_5X_2+c_5X_3+n_3X_4,$} &  {\scriptsize$c_5\in \mathbb R-\{0\}.$} \\
	
	\hline
	{\scriptsize $$}& {\scriptsize  $n(X_1)=n_2X_1, $}&	{\scriptsize  $n(X_2)=n_3X_2,$} &{\scriptsize  $c_4 \in \mathbb R.$}\\
	
	{\scriptsize $$}& {\scriptsize  $n(X_3)=c_4X_2+n_2X_3,  $}&	{\scriptsize  $n(X_4)=(n_2-n_3-c_4)X_1+n_3X_4,$} &  {\scriptsize$$} \\
	
	\hline
	
	{\scriptsize $A_{4,5}^{a,-a}$}& {\scriptsize  $n(X_1)=n_1X_1, $}&	{\scriptsize  $n(X_2)=-n_2X_1+n_4X_2,$} &{\scriptsize  $n_2,n_3\neq 0.$}\\
	
	{\scriptsize $$}& {\scriptsize  $n(X_3)=n_3X_1+n_4X_3,  $}&	{\scriptsize  $n(X_4)=n_3X_2+n_2X_3+n_1X_4,$} &  {\scriptsize$\ast$} \\
	
	\hline
	
	{\scriptsize $$}& {\scriptsize  $n(X_1)=n_1X_1, $}&	{\scriptsize  $n(X_2)=n_4X_2,$} &{\scriptsize  $c_3 \in \mathbb R-\{0\}.$}\\
	
	{\scriptsize $$}& {\scriptsize  $n(X_3)=c_3X_1+n_4X_3,  $}&	{\scriptsize  $n(X_4)=c_3X_2+n_1X_4,$} &  {\scriptsize$$} \\
	
	\hline
	{\scriptsize $$}& {\scriptsize  $n(X_1)=n_1X_1, $}&	{\scriptsize  $n(X_2)=-c_2X_1+n_4X_2,$} &{\scriptsize  $c_2 \in \mathbb R-\{0\}.$}\\
	
	{\scriptsize $$}& {\scriptsize  $n(X_3)=n_4X_3,  $}&	{\scriptsize  $n(X_4)=c_2X_3+n_1X_4,$} &  {\scriptsize$$} \\
	
	\hline
	{\scriptsize $$}& {\scriptsize  $n(X_1)=n_1X_1, $}&	{\scriptsize  $n(X_2)=n_4X_2,$} &{\scriptsize  $n_1\neq n_4.$}\\
	
	{\scriptsize $$}& {\scriptsize  $n(X_3)=n_4X_3,  $}&	{\scriptsize  $n(X_4)=n_1X_4,$} &  {\scriptsize$$} \\
	
	\hline
		{\scriptsize $A_{4,5}^{-1,a}$}& {\scriptsize  $n(X_1)=n_1X_1+n_3X_3, $}&	{\scriptsize  $n(X_2)=n_1X_2+n_2X_3,$} &{\scriptsize  $n_2,n_3\neq 0.$}\\
		
		{\scriptsize $$}& {\scriptsize  $n(X_3)=n_4X_3,  $}&	{\scriptsize  $n(X_4)=n_2X_1-n_3X_2+n_4X_4,$} &  {\scriptsize$\ast$} \\
		
		\hline
		
		{\scriptsize $$}& {\scriptsize  $n(X_1)=n_1X_1+c_3X_3, $}&	{\scriptsize  $n(X_2)=n_1X_2,$} &{\scriptsize  $c_3\in \mathbb R-\{0\}.$}\\
		
		{\scriptsize $$}& {\scriptsize  $n(X_3)=n_4X_3,  $}&	{\scriptsize  $n(X_4)=-c_3X_2+n_4X_4,$} &  {\scriptsize$$} \\
		
		\hline
		{\scriptsize $$}& {\scriptsize  $n(X_1)=n_1X_1, $}&	{\scriptsize  $n(X_2)=n_1X_2+c_2X_3,$} &{\scriptsize  $c_2\in \mathbb R-\{0\}.$}\\
		
		{\scriptsize $$}& {\scriptsize  $n(X_3)=n_4X_3,  $}&	{\scriptsize  $n(X_4)=c_2X_1+n_4X_4,$} &  {\scriptsize$$} \\
		
		\hline
		{\scriptsize $$}& {\scriptsize  $n(X_1)=n_1X_1, $}&	{\scriptsize  $n(X_2)=n_1X_2,$} &{\scriptsize  $n_1\neq n_4.$}\\
		
		{\scriptsize $$}& {\scriptsize  $n(X_3)=n_4X_3,  $}&	{\scriptsize  $n(X_4)=n_4X_4,$} &  {\scriptsize$$} \\
		
		\hline
			{\scriptsize $II\oplus\mathbb R$}& {\scriptsize  $n(X_1)=n_1X_1+n_5X_4, $}&	{\scriptsize  $n(X_2)=-c_2X_1+n_4X_2+n_5X_3,$} &{\scriptsize  $n_5\neq 0, n_1\neq n_4,$}\\
			
			{\scriptsize $$}& {\scriptsize  $n(X_3)=n_3X_2+n_1X_3+c_2X_4,  $}&	{\scriptsize  $n(X_4)=n_3X_1+n_4X_4,$} &  {\scriptsize$c_2\in \mathbb R.$} \\
			
			\hline
	{\scriptsize $$}& {\scriptsize  $n(X_1)=n_1X_1+n_5X_4, $}&	{\scriptsize  $n(X_2)=-c_2X_1+n_1X_2+n_5X_3,$} &{\scriptsize  $n_3,n_5\neq 0,$}\\
	
	{\scriptsize $$}& {\scriptsize  $n(X_3)=n_3X_2+n_1X_3+c_2X_4,  $}&	{\scriptsize  $n(X_4)=n_3X_1+n_1X_4,$} &  {\scriptsize$c_2\in \mathbb R.$} \\
			\hline
			
	{\scriptsize $$}& {\scriptsize  $n(X_1)=n_1X_1+n_5X_4, $}&	{\scriptsize  $n(X_2)=-c_2X_1+n_1X_2+n_5X_3,$} &{\scriptsize  $n_5\neq 0,$}\\
			
	{\scriptsize $$}& {\scriptsize  $n(X_3)=n_1X_3+c_2X_4,  $}&	{\scriptsize  $n(X_4)=n_1X_4,$} &  {\scriptsize$c_2\in \mathbb R.$} \\			
			\hline

	{\scriptsize $$}& {\scriptsize  $n(X_1)=n_1X_1, $}&	{\scriptsize  $n(X_2)=-c_2X_1+n_4X_2,$} &{\scriptsize  $\quad n_1\neq n_4,$}\\
	
	{\scriptsize $$}& {\scriptsize  $n(X_3)=c_3X_2+n_1X_3+c_2X_4,  $}&	{\scriptsize  $n(X_4)=c_3X_1+n_4X_4,$} &  {\scriptsize$c_2,c_3\in \mathbb R.$} \\
	\hline
	{\scriptsize $$}& {\scriptsize  $n(X_1)=n_1X_1, $}&	{\scriptsize  $n(X_2)=-c_2X_1+n_1X_2,$} &{\scriptsize  $ c_2\in \mathbb R,$}\\
	
	{\scriptsize $$}& {\scriptsize  $n(X_3)=c_3X_2+n_1X_3+c_2X_4,  $}&	{\scriptsize  $n(X_4)=c_3X_1+n_1X_4,$} &  {\scriptsize$c_3\in \mathbb R-\{0\}.$} \\
	\hline	
	{\scriptsize $$}& {\scriptsize  $n(X_1)=n_1X_1, $}&	{\scriptsize  $n(X_2)=-c_2X_1+n_1X_2,$} &{\scriptsize  $ c_2\in \mathbb R-\{0\}.$}\\
	
	{\scriptsize $$}& {\scriptsize  $n(X_3)=n_1X_3+c_2X_4,  $}&	{\scriptsize  $n(X_4)=n_1X_4,$} &  {\scriptsize$$} \\
	\hline	
	
\end{tabular}
{\footnotesize   \bf{Table 3.}} \label{TT}
{\footnotesize (Continued.) }\\   \begin{tabular}{l | l l l l lp{40mm} }
	
	\hline\hline
	{\scriptsize $\mathfrak g$ }&{\scriptsize  $\mbox{Equivalence classes}$}&{\scriptsize $$}
	& {\scriptsize  $$}
	\smallskip\\
	\hline
	\smallskip
			
	{\scriptsize $VI_0\oplus \mathbb R$}& {\scriptsize  $n(X_1)=n_1X_1+n_3X_4, $}&	{\scriptsize  $n(X_2)=n_1X_2+n_2X_4,$} &{\scriptsize  $ n_2, n_3\neq 0,$}\\
	
	{\scriptsize $$}& {\scriptsize  $n(X_3)=n_2X_1-n_3X_2+n_4X_3,  $}&	{\scriptsize  $n(X_4)=n_4X_4,$} &  {\scriptsize$\ast$} \\		
	
	\hline
	
	{\scriptsize $$}& {\scriptsize  $n(X_1)=n_1X_1+c_3X_4, $}&	{\scriptsize  $n(X_2)=n_1X_2,$} &{\scriptsize  $ c_3\in \mathbb R-\{0\}.$}\\
	
	{\scriptsize $$}& {\scriptsize  $n(X_3)=-c_3X_2+n_4X_3,  $}&	{\scriptsize  $n(X_4)=n_4X_4,$} &  {\scriptsize$$} \\	
	\hline	
		{\scriptsize $$}& {\scriptsize  $n(X_1)=n_1X_1, $}&	{\scriptsize  $n(X_2)=n_1X_2+c_2X_4,$} &{\scriptsize  $ c_2\in \mathbb R-\{0\}.$}\\
		
		{\scriptsize $$}& {\scriptsize  $n(X_3)=c_2X_1+n_4X_3,  $}&	{\scriptsize  $n(X_4)=n_4X_4,$} &  {\scriptsize$$} \\	
		\hline
		{\scriptsize $$}& {\scriptsize  $n(X_1)=n_1X_1, $}&	{\scriptsize  $n(X_2)=n_1X_2,$} &{\scriptsize  $ n_1\neq n_4.$}\\
		
		{\scriptsize $$}& {\scriptsize  $n(X_3)=n_4X_3,  $}&	{\scriptsize  $n(X_4)=n_4X_4,$} &  {\scriptsize$$} \\	
		\hline	
	{\scriptsize $VII_0\oplus \mathbb R$}& {\scriptsize  $n(X_1)=n_1X_1+n_3X_4, $}&	{\scriptsize  $n(X_2)=n_1X_2+n_2X_4,$} &{\scriptsize  $ n_2\;\mbox{or}\; n_3\neq 0,$}\\
	
	{\scriptsize $$}& {\scriptsize  $n(X_3)=n_2X_1-n_3X_2+n_4X_3,  $}&	{\scriptsize  $n(X_4)=n_4X_4,$} &  {\scriptsize$\ast$} \\	
		
		\hline 
		
	{\scriptsize $$}& {\scriptsize  $n(X_1)=n_1X_1, $}&	{\scriptsize  $n(X_2)=n_1X_2,$} &{\scriptsize  $ n_1\neq  n_4.$}\\
	
	{\scriptsize $$}& {\scriptsize  $n(X_3)=n_4X_3,  $}&	{\scriptsize  $n(X_4)=n_4X_4,$} &  {\scriptsize$$} \\		
	\hline
	\hline		
	\end{tabular}

Note that, the Nijenhuis structures denoted by $(\ast)$ on the Lia algebras $A^{a,-a}_{4,5}$, $A^{-1,a}_{4,5}$ and $VI_0\oplus \mathbb R$ for which the multiplication $n_2n_3$ are the same value, are equivalent; and on the Lie algebra $VII_0\oplus \mathbb R$ the structures for which $n_2n_3=0$, are equivalent.

\section{Physical application}\label{app}
 We consider the method of constructing classical dynamical systems by using $r$-matrices (see \cite{Zh},\cite{ReSe},\cite{Egh}) for $r$-$n$ structures.

 Let manifold $M^{2n}$ with the symplectic structure $\omega_{ij}$ and local coordinate $\{x_i\}$ be as a phase space and let Lie group $G$ with Lie algebra $\mathfrak g$ be symmetry group of a dynamical system. The Poisson bracket $\{\cdot,\cdot\}$ defined by symplectic structure $\omega_{ij}$ for arbitrary functions $f,g\in C^{\infty}(M)$ is
 \[
 \{f,g\}=\Pi^{ij}\frac{\partial f}{\partial x_i}\frac{\partial g}{\partial x_j},
 \]
 where $\Pi^{ij}$ is the inverse of the matrix $\omega_{ij}$.

 \noindent Now, we consider a dynamical system with symmetry group $G$ for which dynamical variables $S_k=S_k(x_i)$, $(k=1,...,dim \mathfrak g)$ are constructed as functions on $M$ satisfying the relation
 $$ \{S_i,S_j\}=C^k_{ij}S_k,$$
 where $C^k_{ij}$ are structure constants of the Lie algebra $\mathfrak g$ of the symmetry group $G$.

  \noindent Consider an invertible $r$-matrix $r\in \wedge ^2\mathfrak g$,  where $r=r^{ij}X_i\wedge X_j$. If we choose an irreducible representation of the Lie algebra $\mathfrak g$ and take $\{T_i\}$ to be a basis of $\mathfrak g$ in this representation, then by $\mathfrak g$-valued functions
\begin{equation}\label{function}
Q(x)=S_ir^{ij}T_j,
\end{equation}
 one may define functions
 \begin{equation} \label{I}
 I_k=trace(Q^k), \quad k\in \mathbb N,
 \end{equation}
as constants of motion of the dynamical system (for more details see, \cite{Zh}). The dynamical system would be completely integrable (in the sense of Liouville) if there are $n$ independent constants of motion and would be superintegrable if there are additional independent invariants up to $2n-1$.

 Now consider an $r$-$n$ structure $(r,n) $ and $r$-matrix $n\circ r$, thus
\[
Q=S_in^i_kr^{kj}T_j\quad \Longrightarrow \quad I_1=S_in^i_kr^{kj}tr(T_j).
\]

\noindent Suppose ${r^{(l)}}=n^{(l)}\circ r$ be other solutions of the CYBE, so
\[
Q^{(l)}=S_i(n^l)^i_kr^{kj}T_j\quad \Longrightarrow \quad I_1^{(l)}=S_i(n^l)^i_kr^{kj}tr(T_j).
\]
\begin{remark}\label{last}
If we take $H:=I_1$ and $H^l:=I_1^{(l)}$ as Hamiltonian functions of dynamical systems, then $H=\sum H^{(l)}$, because we can take $n:=\sum n^{(l)}$.
\end{remark}
\subsection*{Example}

Consider the Euclidean space $\mathbb R^4$ as a phase space with the standard symplectic structure
 $$\omega=dx_1\wedge dx_2+dx_3\wedge dx_4,$$
  so the Poisson bracket is
\[
\{x_1,x_2\}=\{x_3,x_4\}=0, \quad \{x_1,x_3\}=\{x_2,x_4\}=1.
\]
Suppose Lie algebra $A_{4,1}$ be a symmetry group of a dynamical system. One can find functions $S_k(x_i)$ by using the realizations of the Lie algebra \cite{Ab}. Realizations of real solvable four-dimensional Lie algebras are listed in table $5$ of \cite{Reali}. We take the realizations
\[
X_1:=\partial _1,\quad X_2:=x_2\partial _1,\quad X_3:=\frac{1}{2}x_2^2\partial _1,\quad X_4:=-\partial _2,
\]
of the Lie algebra $A_{4,1}$ in that table, thus functions $S_k(x_i)$ can be regarded as
\[
S_1=-x_3,\quad S_2=-x_2x_3,\quad S_3=-\frac{1}{2}x_2^2x_3,\quad S_4=x_4.\quad
\]
\noindent Since Lie algebra $A_{4,1}$ is solvable, we can take a representation of four-dimensional triangular matrices for this Lie algebra as follows:
\[
 {\scriptsize
 	T_{1}=\left(\begin{array}{cccc}
 0& 1& 0& 1 \\ 0& 0& 0& 0  \\ 0& 0& 0& 0  \\0& 0& 0& 1\\
 	\end{array}\right),\quad
 	T_{2}=\left(\begin{array}{cccc}
  0& 1& 0& 2 \\ 0& 0& 0& 0  \\ 0& 0& 0& 0  \\0& 0& 0& 0\\ \end{array} \right),\quad
 T_{3}=\left(\begin{array}{cccc}
 1& 1& 0& 0 \\ 0& 1& 0& 1 \\ 0& 0& 0& 0  \\0& 0& 0& 0\\ \end{array} \right),\quad
T_{4}=\left(\begin{array}{cccc}
 0& 0& 0& 1 \\ 0& 1& 0& 1  \\ 0& 0& 0& 0  \\0& 0& 0& 0\\ \end{array} \right)},
 \]
 where $[T_i,T_j]=C_{ij}^kT_k$ and $C^k_{ij}$ are structure constants of Lie algebra $A_{4,1}$. From (\ref{function}) for the r-matrix $r=X_1 \wedge X_4-X_2 \wedge X_3$ we have
\[
Q(x)=-S_4T_1+S_3T_2-S_2T_3+S_1T_4.
\]
Using (\ref{I}), constants of motion are
\[
I_1 = 2x_2x_3-x_3-x_4,\quad 
I_2=x_2^2x_3^2+(x_2x_3-x_3)^2+x_4^2,\quad
I_3= x_2^3x_3^3+(x_2x_3-x_3)^3-x_4^3,
\]
which means the dynamical system is superintegrable. We consider the following representatives of other equivalence classes of $r$-matrices for this algebra
\[
r^{(1)}= X_1 \wedge X_2 \quad  r^{(2)}=X_3 \wedge X_1 ,\quad
r^{(3)}= X_1 \wedge X_4-X_1\wedge X_3,
\quad r^{(4)}= X_3 \wedge X_2.
\]
One can find functions $Q^{(l)}$ as
\[
\begin{array}{rclcrclcrclcrcl}
Q^{(1)}&=& S_1T_2-S_2T_1,&&
Q^{(2)}&=& S_3T_1-S_1T_3,\\[4pt]
Q^{(3)}&=& (S_3-S_4)T_1-S_1T_3+S_1T_4,&&
Q^{(4)}&=& S_3T_2-S_2T_3,\\
\end{array}
\]
for the solutions $r^{(1)}$, $r^{(2)}$, $r^{(3)}$ and $r^{(4)}$, respectively; their corresponding constants of motion are
\[
\begin{array}{rclcrclcrcl}
I_1^{(1)}&=& x_2x_3,\\[6pt]
I_1^{(2)}&=& 2x_3-\frac{1}{2}x_2^2x_3,&&
I_2^{(2)}&=& 2x_3^2+(\frac{1}{2}x_2^2x_3)^2,\\[6pt]
I_1^{(3)}&=& x_3-\frac{1}{2}x_2^2x_3-x_4,&&
I_2^{(3)} &=& x_3^2+(\frac{1}{2}x_2^2x_3+x_4)^2,\\[6pt]
I_1^{(4)}&=& 2x_2x_3.\\
\end{array}
\]

\noindent They are constants of motion of four dynamical systems which have the same symmetry group $A_{4,1}$. Dynamical systems related to the second and third one are completely integrable.

Since $r$-matrices $r$ and $r^{(l)}$ are compatible, there exist Nijenhuis structures $n^{(l)}:=r^{(l)}\circ r^{-1}$ such that the couples $(r,n^{(l)})$ are compatible $r$-$n$ structures. In fact, $n^{(l)}$'s are the following elements of the equivalence classes of Nijenhuis structures given in table $3$ corresponding to $r$-$n$ structures $(r,n)$ on the Lie algebra $A_{4,1}$
\[
\begin{array}{rclcrclcrclcrclcrclcrclcrclcrcl}
n^{(1)}(X_1)&=& 0, && n^{(1)}(X_2)&=& 0,&& n^{(1)}(X_3)&=&  X_1, && n^{(1)}(X_4)&=&  X_2,\\

n^{(2)}(X_1)&=&  0,&& n^{(2)}(X_2)&=&  X_1,&& n^{(2)}(X_3)&=&  0, && n^{(2)}(X_4)&=&  -X_3,\\

n^{(3)}(X_1)&=&  X_1, && n^{(3)}(X_2)&=& X_1,&& n^{(3)}(X_3)&=&  0, && n^{(3)}(X_4)&=&  -X_3+X_4,\\
n^{(4)}(X_1)&=&  0, && n^{(4)}(X_2)&=& X_2,&& n^{(4)}(X_3)&=&  X_3, && n^{(4)}(X_4)&=&  0.\\
\end{array}
\] 

\noindent  According to Remark \ref{last}, if we take $H^{l}:=I^{(l)}_1$ as Hamiltonian functions of dynamical systems, then $$H:=\sum H^{(l)}:=I_{1}=3x_3+3x_2x_3-x_2^2x_3-x_4$$
 is a Hamiltonian function for the dynamical system which corresponds to $r$-matrix $n\circ r$, such that $n:=\sum n^{(l)}$ is
\[
n(X_1)=X_1,\quad n(X_2)=2X_1+X_2, \quad n(X_3)=X_1+X_3,\quad n(X_4)=X_2-2X_3+X_4.
\]
This dynamical system is superintegrable and two others constants of motion are
\[
\begin{array}{rcl}
I_2 &:=& (x_2x_3+2x_3)^2+(x_2x_3+x_3)^2+(x_2x_3-x_4-x_2^2x_3)^2,\\[3pt]
I_3 &:=&(x_2x_3+2x_3)^3+(x_2x_3+x_3)^3+(x_2x_3-x_4-x_2^2x_3)^3.
\end{array}
\]


\section{Concluding remarks}
In this paper we consider invariant Poisson-Nijenhuis structures on Lie groups. We characterized the infinitesimal counterpart of such structures and showed that, under a certain condition, they are in one-to-one correspondence with compatible solutions of the classical Yang-Baxter equation. As an important result, an invertible $r$-matrix can be used to produce compatible $r$-matrices by using the theory of $r$-$n$ structures.

 Given a phase space with a Lie group as a symmetry group of a dynamical system, using the results of Section \ref{Section4} and the classification procedure in Section \ref{method}, one can explicitly find a number of dynamical systems on phase space which can be constructed using $r$-matrices.

We listed all $r$-matrices and all $r$-$n$ structures with invertible $r$, up to a natural equivalence, on four-dimensional symplectic real Lie algebras; equivalently we have, up to an equivalence, all invariant Poisson structures, and all invariant Poisson-Nijenhuis structures with invertible Poisson structures on the corresponding Lie groups.

It is wroth mentioning that having right invariant Poisson structures, one can produce multiplicative Poisson structures which are the Poisson structures for which the group multiplication is a Poisson epimorphism \cite{Va}. Under the same hypotheses as in Corollary \ref{hierarchy-right-inv}, we deduce that  
\[
\Pi_k = \lvec{r_k} - \rvec{r_k}
\]
is a coboundary multiplicative Poisson structure on $G$, for every $k \in \mathbb{N}$. Moreover, since $\rvec{r_k}$ and $\rvec{r_l}$ are compatible right-invariant Poisson structures, it follows that $[r_k, r_l] = 0$ and, thus,
\[
[\Pi_k, \Pi_l] = \lvec{[r_k, r_l]} + \rvec{[r_k, r_l]} = 0
\]
(note that $[\lvec{r_k}, \rvec{r_l}] = 0$). Therefore, $\Pi_k$ and $\Pi_l$ are compatible multiplicative Poisson structures on $G$.

A generalization of the Poisson coalgebra approach to bi-Hamiltonian systems was shown in \cite {BaMaRa}, and a construction of integrable deformations of bi-Hamiltonian systems, based on the theory of multiplicative Poisson structures on the Lie groups, was presented (for more details, see \cite{BaBlMu}, \cite{BaMaRa}, \cite{BaRa}).

 As an application of our work, one can investigate the deformation of completely integrable Hamiltonian and bi-Hamiltonian systems using the hierarchies of multiplicative Poisson structures associated with right-invariant $P$-$N$ structures on the Lie groups, and the relation between the obtained results on the Lie algebras in the present paper and the underlying geometric structure of such systems on the Lie groups.
  
 It would be also of interest to get insight into the relation between $r$-$n$ structures and quantum $r$-matrices and study the procedure along quantization, we hope to present this elsewhere.

 \subsection*{Acknowledgments}

 The first author is greatly indebted to Professor Juan Carlos Marrero for his support and valuable discussions and remarks. She would also like to express her gratitude to University of La Laguna for the hospitality in a visit. Her warmest thanks to Professor Janusz Grabowski, for his patience at carefully reading the manuscript and helpful suggestions to improve significantly the presentation of the paper, and Institute of Mathematics Polish Academy of Sciences for the hospitality in a visit where a part of this work was done. She also would like to thank Professor Angel Ballesterose for his useful comment.

\subsection*{Appendix}
\begin{center}
{\footnotesize   Table 4.} {\footnotesize  
 	Automorphisms groups of four-dimensional symplectic real Lie algebras.}\\
 \end{center}
\begin{tabular}{| l | l | l | l | p{15mm} }
\hline\hline
\vspace{-1mm}
{\scriptsize Lie algebra }&{\scriptsize Automorphisms group}
& {\scriptsize Lie algebra }& {\scriptsize Automorphisms group}  \smallskip\\
\hline


{\scriptsize $II \oplus \mathbb R$} & {\scriptsize $\left(\begin{array}{cccc}
	a_{11}a_6 - a_{10}a_7 & a_2& a_3& a_4\\
	0& a_6& a_7& 0\\   0& a_{10}& a_{11} & 0\\   0& a_{14}& a_{15}& a_{16}\\
	\end{array} \right)$}&

{\scriptsize $VI_{0}\oplus \mathbb R$}&

{\scriptsize $\left(\begin{array}{cccc}
	a_1 & 0& a_3& 0\\
	0& a_6& a_7& 0\\   0& 0& 1 & 0\\   0& 0& a_{15}& a_{16}\\
	\end{array} \right)$}\\

\hline


{\scriptsize $VII_{0}\oplus \mathbb R$} & {\scriptsize $\left(\begin{array}{cccc}
	a_1 & a_2& a_3& 0\\
	-a_2& a_1& a_7& 0\\   0& 0& 1 & 0\\   0& 0& a_{15}& a_{16}\\
	\end{array} \right)$}&

{\scriptsize $A_2\oplus A_2$}&
{\scriptsize $\left(\begin{array}{cccc}
	1 & 0&  0&  0\\
	a_5 & a_6& 0& 0\\0 & 0& 1& 0\\0 & 0& a_{15} & a_{16}\\
	\end{array} \right)$}\\
\hline


{\scriptsize $A_{4,2}^{-1}$}&

{\scriptsize $\left(\begin{array}{cccc}
	a_1& 0& 0 &a_4\\
	0& a_6& 0& a_8\\    0& 0& a_6& 0\\  0& 0 & 0& 1\\
	\end{array} \right)$}&

{\scriptsize $A_{4,3}$} & {\scriptsize $\left(\begin{array}{cccc}
	a_1& 0& 0& a_4\\
	0& a_6 & a_7& a_8\\  0 & 0& a_6& a_{12}\\    0& 0& 0 & 1\\
	\end{array} \right)$}\\

\hline


{\scriptsize $A_{4,5}^{a,-a}$} & {\scriptsize $\left(\begin{array}{cccc}
	a_1& 0& 0& a_4\\
	0& a_6& 0& a_8\\  0 & 0& a_{11}& a_{12}\\    0& 0& 0 & 1\\
	\end{array} \right)$}&

{\scriptsize $A_{4,5}^{-1,-1}$} & {\scriptsize $\left(\begin{array}{cccc}
	a_1& 0& 0& a_4\\
	0& a_6& a_7& a_8\\  0 & a_{10}& a_{11}& a_{12}\\    0& 0& 0 & 1\\
	\end{array} \right)$}\\

\hline


{\scriptsize $A_{4,6}^{a,0}$} & {\scriptsize $\left(\begin{array}{cccc}
	a_1& 0& 0& a_4\\
	0& a_6 & a_7& a_8\\  0 & -a_7& a_6& a_{12}\\   0& 0& 0 & 1\\
	\end{array} \right)$}&

{\scriptsize $A_{4,5}^{-1,a}$} & {\scriptsize $\left(\begin{array}{cccc}
	a_1& 0& 0& a_4\\
	0& a_6& 0& a_8\\  0 & 0& a_{11}& a_{12}\\    0& 0& 0 & 1\\
	\end{array} \right)$}\\

\hline


{\scriptsize $A_{4,9}^{-\frac{1}{2}}$} & {\scriptsize $\left(\begin{array}{cccc}
	a_{11}a_6& 2a_{12}a_6& a_8a_{11}& a_{4}\\
	0& a_6& 0& a_8\\ 0& 0& a_{11}& a_{12}\\  0& 0& 0& 1\\
	\end{array} \right)$}&

{\scriptsize $A_{4,1}$}&
{\scriptsize $\left(\begin{array}{cccc}
	a_{11}	a_{16}^{2}& a_7a_{16}& a_3& a_4\\
	0 & a_{11}a_{16} & a_7& a_8\\0 & 0& a_{11}& a_{12}\\0& 0& 0& a_{16}\\
	\end{array} \right)$}\\

\hline


{\scriptsize $A_{4,9}^{0}$} & {\scriptsize $\left(\begin{array}{cccc}
	a_{11}a_6 & a_2& a_8a_{11}& a_4\\
	0 & a_6 & 0& a_8\\ 0& 0& a_{11}& 0\\ 0& 0& 0& 1\\
	\end{array} \right)$}&

{\scriptsize $A_{4,9}^{b}$} & {\scriptsize $\left(\begin{array}{cccc}
	a_{11}a_6 & -\displaystyle\frac{a_{12}a_6}{b}& a_8a_{11}& a_4\\
	0 & a_6 & 0& a_8\\ 0& 0& a_{11}& a_{12}\\ 0& 0& 0& 1\\
	\end{array} \right)$}\\

\hline


{\scriptsize $A_{4,12}$} & {\scriptsize $\left(\begin{array}{cccc}
	a_1& a_2& a_3& a_4\\
	-a_2& a_1& a_4& -a_3\\ 0& 0& 1& 0\\  0& 0& 0& 1\\
	\end{array} \right)$}&
{\scriptsize $A_{4,7}$}&

{\scriptsize $\left(\begin{array}{cccc}
	a_{6}^{2}& -a_{12}a_6& a_6a_8-a_{12}(a_6+a_7) & a_4\\
	0& a_6& a_7& a_8\\   0& 0& a_6& a_{12}\\  0& 0 & 0& 1\\
	\end{array} \right)$}\\
\hline


{\scriptsize $A_{4,11}^{b}$}&
{\scriptsize $\left(\begin{array}{cccc}
	a_{6}^{2}+ a_{7}^{2}& \displaystyle\frac{-A}{1+b^2}&  \displaystyle\frac{-B}{1+b^2} & a_4\\ 
	0& a_6& a_7& a_8\\
	
	0&	-a_7& a_6& a_{12}\\  0& 0& 0& 1 \\
	\end{array} \right)$}&

{\scriptsize $A_{4,9}^{1}$}&
{\scriptsize $\left(\begin{array}{cccc}
	a_{11}a_6 - a_{10}a_7& C& D& a_4\\ 0& a_6& a_7& a_8\\
	0& a_{10}& a_{11} & a_{12}\\   0& 0& 0& 1\\
	\end{array} \right)$}\\

{\scriptsize $$}& {\scriptsize$A:=a_6(ba_{12}+a_8)+a_7(ba_8-a_{12})$}&{\scriptsize $$}&{\scriptsize $C:=a_{10}a_8-a_{12}a_6$}\\
{\scriptsize $$}& {\scriptsize$B:=a_6(a_{12}-ba_8)+a_7(ba_{12}+a_{8})$}&{\scriptsize $$}& {\scriptsize$D:=a_8a_{11}-a_7a_{12}$}\\

\hline
\end{tabular}

 \end{document}